\documentclass[12pt]{article}
\usepackage{amsthm,amssymb,amsfonts,amsmath,amscd}
\usepackage{graphicx}
\makeatletter \@addtoreset{equation}{section} \makeatother
\newtheorem{proposition}{Proposition}
\newtheorem{theorem}{Theorem}
\newtheorem{lemma}{Lemma}

\def\dfrac{\displaystyle\frac}
\def\sumd{\displaystyle\sum}

\def\intd{\displaystyle\int}

\begin{document}

\title{On Universality of Bulk Local Regime of the Deformed Gaussian
Unitary Ensemble}
\author{ T. Shcherbina\\
 Institute for Low Temperature Physics, Kharkov,
Ukraine. \\E-mail: t\underline{ }shcherbina@rambler.ru
}

\date{}

\maketitle

\begin{abstract}
We consider the deformed Gaussian Ensemble $H_n=M_n+H^{(0)}_n$ in which $%
H_n^{(0)}$ is a diagonal Hermitian matrix and $M_n$ is the Gaussian Unitary
Ensemble (GUE) random matrix. Assuming that the Normalized Counting Measure
of $H_n^{(0)}$ (both non-random and random) converges weakly to a measure $%
N^{(0)}$ with a bounded support we prove universality of the local
eigenvalue statistics in the bulk of the limiting spectrum of
$H_n$.
\end{abstract}

\section{Introduction.}

Universality is an important topic of the random matrix theory. It
deals with statistical properties of eigenvalues of $n\times n$
random matrices on intervals whose length tends to zero as $n \to
\infty$. According to the universality hypothesis these properties
do not depend to large extent on the ensemble. The hypothesis was
formulated in the early 60s and since then was proved in certain
cases.
There
are some results only for special cases. Best of all universality
is studied in the case of ensembles with a unitary invariant
probability distribution (known also as unitary matrix models)
(\cite{De-Co:99, Pa-Sh:97, Pa-Sh:07}).

To formulate the universality hypothesis we need some notations
and definitions. Denote by
$\lambda_1^{(n)},\ldots,\lambda_n^{(n)}$ the eigenvalues of the
random matrix. Define the normalized eigenvalue counting measure
(NCM) of the matrix as
\begin{equation}  \label{NCM}
N_n(\triangle)=\sharp\{\lambda_j^{(n)}\in
\triangle,j=\overline{1,n} \}/n,\quad N_n(\mathbb{R})=1,
\end{equation}
where $\triangle$ is an arbitrary interval of the real axis. For many known
random matrices the expectation $\overline{N}_n=\mathbf{E}\{N_n\}$ is
absolutely continuous, i.e.,
\begin{equation}  \label{rhon}
\overline{N}_n(\triangle)=\displaystyle\int\limits_\triangle \rho_n
(\lambda)d\,\lambda.
\end{equation}
The non-negative function $\rho_n$ in (\ref{rhon}) is called the density of
states.

Define also the $m$-point correlation function $R_m^{(n)}$ by the
equality:
\begin{equation}  \label{R}
\mathbf{E}\left\{ \sum_{j_{1}\neq ...\neq j_{m}}\varphi_m
(\lambda_{j_{1}},\dots,\lambda_{j_{m}})\right\} =\int \varphi _{m}
(\lambda_{1},\dots,\lambda_{m})R_{m}^{(n)}(\lambda_{1},\dots,\lambda_{m})
d\lambda_{1},\dots,d\lambda_{m},
\end{equation}
where $\varphi_{m}: \mathbb{R}^{m}\rightarrow \mathbb{C}$ is
bounded, continuous and symmetric in its arguments and the
summation is over all $m$-tuples of distinct integers $
j_{1},\dots,j_{m}=\overline{1,n}$. Here and below integrals
without limits denote the integration over the whole real axis.

The global regime of the random matrix theory, centered around
weak convergence of the normalized counting measure of
eigenvalues, is well-studied for many ensembles. It is shown that
$N_n$ converges weakly to a non-random limiting measure $N$ known
as the integrated density of states (IDS). The IDS is normalized
to unity and is absolutely continuous
\begin{equation}  \label{rho}
N(\mathbb{R})=1,\quad N(\triangle)=\displaystyle\int\limits_\triangle
\rho(\lambda)d\,\lambda.
\end{equation}
The non-negative function $\rho$ in (\ref{rho}) is called the limiting
density of states of the ensemble.

We will call the spectrum the support of $N$ and define the bulk
of the spectrum as
\begin{equation}
\hbox{bulk}\,N=\{\lambda|\exists (a,b)\subset \hbox{supp}\,N:\lambda\in
(a,b),\,\, \rho_n(\mu)\rightrightarrows \rho(\mu)\,\hbox{on}\, (a,b),
\rho(\lambda)\not=0\}.
\end{equation}
Then the universality hypothesis on the bulk of the spectrum says
that for $ \lambda_0\in\hbox{bulk}\, N$ we have:

(i) for any fixed $m$ uniformly in $x_1, x_2,\ldots, x_m$ varying
in any compact set in $\mathbb{R}$
\begin{equation}  \label{Un}
\lim\limits_{n\to \infty}\displaystyle\frac{1}{(n\rho_n(\lambda_0))^m}
R^{(n)}_m\left(\lambda_0+\displaystyle\frac{x_1}{\rho_n(\lambda_0)\,n},
\ldots,\lambda_0+\displaystyle\frac{x_m}{\rho_n(\lambda_0)\,n}\right)=\det
\{S(x_i-x_j)\}_{i,j=1}^m,
\end{equation}
where
\begin{equation}  \label{S}
S(x_i-x_j)=\displaystyle\frac{\sin \pi(x_i-x_j)}{\pi(x_i-x_j)},
\end{equation}
and $R^{(n)}_m$, $\rho_n$, and $\rho$ are defined in (\ref{R}),
(\ref{rhon}) and (\ref{rho});

(ii) if
\begin{equation}
E_{n}(\triangle )=\mathbf{P}\{\lambda _{i}^{(n)}\not\in \triangle
,\,i= \overline{1,n}\},  \label{gapp}
\end{equation}
is the gap probability, then
\begin{equation}
\lim\limits_{n\rightarrow \infty }E_{n}\left( \left[\lambda
_{0}+\displaystyle \frac{a}{\rho _{n}(\lambda _{0})\,n},\lambda
_{0}+\displaystyle\frac{b}{\rho _{n}(\lambda
_{0})\,n}\right]\right) =\det \{1-S_{a,b}\},  \label{gp}
\end{equation}
where the operator $S_{a,b}$ is defined on $L_{2}[a,b]$ by the formula
\begin{equation*}
S_{a,b}f(x)=\displaystyle\int\limits_{a}^{b}S(x-y)f(y)d\,y,
\end{equation*}
and $S$ is defined in (\ref{S}).

In this paper we study universality of the local bulk regime of
random matrices of the deformed Gaussian Unitary Ensemble (GUE)
\begin{equation}
H_{n}=M_{n}+H_{n}^{(0)},  \label{H}
\end{equation}
where $H_{n}^{(0)}$ is a Hermitian matrix with eigenvalues $
\{h_{j}^{(n)}\}_{j=1}^{n}$ and $M_{n}$ is the GUE matrix, defined
as
\begin{equation}
M_{n}=n^{-1/2}W_{n},  \label{M}
\end{equation}
where $W_{n}$ is a Hermitian $n\times n$ matrix whose elements $\Re w_{jk}$
and $\Im w_{jk}$ are independent Gaussian random variables such that
\begin{equation}
\mathbf{E}\{\Re w_{jk}\}=\mathbf{E}\{\Im w_{jk}\}=0,\quad \mathbf{E}\{\Re
^{2}w_{jk}\}=\mathbf{E}\{\Im ^{2}w_{jk}\}=\displaystyle\frac{1}{2}\quad
(j\neq k),\quad \mathbf{E}\{w_{jj}^{2}\}=1.  \label{W}
\end{equation}
Let
\begin{equation*}
N_{n}^{(0)}(\triangle )=\sharp \{h_{j}^{(n)}\in \triangle
,j=\overline{1,n} \}/n.
\end{equation*}
be the Normalized Counting Measure of eigenvalues of $H_{n}^{(0)}$.

Note also that since the probability law of $M_{n}$ is unitary invariant we
can assume without loss of generality that $H_{n}^{(0)}$ is diagonal.

The global regime for the ensemble (\ref{H})-(\ref{W}) is well
enough studied. In particular, it was shown in \cite{Pa:72} that
if $N_{n}^{(0)}$ converges weakly with probability 1 to a
non-random measure $N^{(0)}$ as $n\rightarrow \infty $, then
$N_{n}$ also converges weakly with probability 1 to a non-random
measure $N$. Moreover the Stieltjes transforms $g$ of $N$ and
$g^{(0)}$ of $ N^{(0)}$ satisfy the equation
\begin{equation*}
g(z)=g^{(0)}(z+g(z)).
\end{equation*}
It follows from the definition (\ref{NCM}) and the above result
that any $n$-independent interval $\Delta $ of spectral axis such
that $N(\Delta )>0$ contains $O(n)$ eigenvalues. Thus, to deal
with a finite number of eigenvalues as $n\rightarrow \infty $, in
particular, with the gap probability, one has to consider spectral
intervals, whose length tends to zero as $n\rightarrow \infty $.
In particular, in the local bulk regime we are about intervals of
the length $O(n^{-1})$.

Random matrix theory posseses a powerful techniques of analysis of
the local regime, based on the so called determinant formulas for
the correlation functions \cite{Me:91}. For the GUE, more general
for the hermitian matrix models, the determinant formulas follow
from the possibility to write the joint probability density of its
eigenvalues as the square of the determinant, formed by certain
orthogonal polynomials and then as the determinant formed by
reproducing kernel of the polynomials, that.are also heavily used
in the subsequent asymptotic analysis \cite
{De-Co:99,Pa-Sh:97,Pa-Sh:07}. Unfortunately, the orthogonal
polynomials \ have not appeared so far in the study of the
deformed Gaussian Unitary Ensemble. However, it was shown in
physical papers \cite{Br-Hi:96,Br-Hi:97,Br-Hi:98} that correlation
functions of the deformed Gaussian Unitary Ensemble can be written
in the determinant form, although the corresponding kernel is not,
in general, a reproducing kernel of a system of orthogonal
polynomials. This was done by using as a crucial step the
Harish-Chandra/Itzykson-Zuber formula for certain integrals over
the unitary group.

This important result was used in \cite{Jo:01}  to prove
universality of the local bulk regime of matrices (\ref{H}), where
$ H_{n}^{(0)}=n^{-1/2}W_{n}$ is a hermitian random matrix with
independent (modulo symmetry) entries:
\begin{eqnarray}
W_{n} &=&\{w_{jk}\}_{j,k}^{n},\;w_{jk}=\overline{w_{kj}}  \label{Jcond} \\
\mathbf{E}\{w_{jk}\} &=&\mathbf{E}\{w_{jk}^{2}\}=0,\;\mathbf{E}
\{|w_{jk}|^{2}\}=1,\;\sup_{j,k}\mathbf{E}\{|w_{jk}|^{p}\}<\infty .
\notag
\end{eqnarray}
It was proved in \cite{Jo:01} that if $p>2(m+2)$, then (\ref{Un}) is valid,
and if $p>6$, then (\ref{gp}) is valid.

Later in the series of the papers \cite{Bl-Ku:04,Ap-Co:05} the
special case of (\ref{H} ) was studied, where $H_{n}^{(0)}$ has
two eigenvalues $\pm a$ of equal multiplicity. In this case
universality in the bulk and at the edge of the spectrum were
proved.

In this paper we consider random matrices (\ref{H}) for a rather
general class of $H_{n}^{(0)}$ both random and nonrandom. The main
results are the following theorems.

\begin{theorem}
\label{thm:1} Let $N_{n}^{(0)}$ be a nonrandom measure that
converges weakly to a measure $N^{(0)}$ with a bounded support.
Then for any $\lambda _{0}\in \emph{bulk}\,N$  the universality
properties (\ref{Un} ) and (\ref{gp}) hold.
\end{theorem}

\begin{theorem}
\label{thm:2} Let the eigenvalues $\{h_{j}^{(n)}\}_{j=1}^{n}$ of
$H_{n}^{(0)} $ in (\ref{H}) be a collection of random variables
independent of $W_{n}$ and such that
$\mathbf{E}^{(h)}\{|h_{j}^{(n)}|^{2}\}<\infty $ (the symbol $
\mathbf{E}^{(h)}\{\ldots \}$ denotes the expectation with respect
to the measure generated by $H_{n}^{(0)}$). Assume that there
exists a non-random measure $N^{(0)}$ of a bounded support such
that  for any finite interval $ \Delta \subset \mathbb{R}$ and for
any $\varepsilon >0$
\begin{equation*}
\lim_{n\rightarrow \infty }\mathbf{P}\{|N^{(0)}(\Delta )-N_{n}^{(0)}(\Delta
)|>\varepsilon \}=0.
\end{equation*}
Then for any $\lambda _{0}\in \emph{bulk}\,N$ the universality
properties (\ref{Un} ) and (\ref{gp}) hold.
\end{theorem}

The paper is organized as follows. In Section $2$ we give a new
proof of determinant formulas for correlation functions (\ref{R})
by the method which is different from those of
\cite{Br-Hi:96,Br-Hi:97}, \cite{Jo:01} and \cite{Bl-Ku:04,
Ap-Co:05}. Namely we use the representation of the resolvent of
the random matrix via the integral with respect to the Grassmann
variables. The integral was introduced by Berezin (see
\cite{Be:87}) and widely used in physics literature (see e.g. book
\cite{Ef:97}). For the reader convenience we give in Appendix a
brief account of the Grassmann integral techniques that we will
use in the paper. Theorem \ref{thm:1} will be proved in Sections
$3-4$. Section $5$ deals with the proof of Theorem~\ref{thm:2}.

\section{The determinant formulas.}

It is well-known (see for example \cite{Me:91}) that the correlation
functions (\ref{R}) for the GUE can be written in the determinant form
\begin{equation}  \label{Det}
R_m^{(n)}(\lambda_1,\ldots,\lambda_m)=\det\{K_n(\lambda_i,\lambda_j)\},
\end{equation}
with
\begin{equation*}
K_n(\lambda_i,\lambda_j)=\sum\limits_{k=0}^{n-1}\phi_k(\lambda_i)\phi_k(
\lambda_j), \quad \phi_k(x)=n^{1/4}h_k(\sqrt{n}x)e^{-nx^2/4},
\end{equation*}
where $\{h_k\}_{k\ge 0}$ are orthonormal Hermite polynomials. We want to
find analogs of these formulas in the case of random matrices (\ref{H}).

\begin{proposition}
Let $H_n$ be the random matrix defined in (\ref{H}) and
$R_m^{(n)}$ be the correlation function (\ref{R}). Then
(\ref{Det}) is valid with
\begin{multline}  \label{K}
K_n(\lambda,\mu) \\
=-n\displaystyle\int\limits_{L}\displaystyle\frac{d\,t}{2\pi}\oint\limits_{C}
\displaystyle\frac{d\,v}{2\pi} \displaystyle\frac{\exp\left\{
-\displaystyle \frac{n}{2}(v^2-2v\lambda-t^2+2\mu\,t))\right\}
}{v-t}\prod\limits_{j=1}^n \left(\displaystyle\frac{t-h_j^{(n)}}
{v-h_j^{(n)}}\right),
\end{multline}
where $L$ is a line parallel to the imaginary axis and lying to
the left of all $\{h_j^{(n)}\}_{j=1}^n$, and the closed contour
$C$ has all $ \{h_j^{(n)}\}_{j=1}^n$ inside and does not intersect
$L$.
\end{proposition}
Representation (\ref{K}) was first obtained in physical papers
\cite{Br-Hi:96, Br-Hi:97}. We obtain this representation by use
the Grassmann integration.
\begin{proof}
Following \cite{Gu:91}, where the GUE was studied, denote
\begin{equation}\label{F1}
F(z_1,z_2,\dots,z_m)={\bf
E}\left\{\prod\limits_{k=1}^m\hbox{Tr}\,\dfrac{1}{H_n-z_k}\right\},
\end{equation}
where $\{z_j\}_{j=1}^m$ are distinct complex numbers, $\Im
z_1=\ldots=\Im z_m=-\varepsilon <0$.
 It is technically easier to study the ratio of the determinants instead of
$\hbox{Tr}\dfrac{1}{H_n-z}$. Denote
\begin{equation}\label{D}
D(z_1,\ldots,z_m;x_1,\ldots,x_m)=\dfrac{\det (H_n-z_1-x_1)\dots
\det (H_n-z_m-x_m)}{\det (H_n-z_1)\dots
\det (H_n-z_m)}.
\end{equation}
Since
\[
-\dfrac{d}{d\,x}\dfrac{\det(H_n-z-x)}{\det(H_n-z)}\Bigg|_{x=0}=\hbox{Tr}\,(H_n-z)^{-1},
\]
then
\begin{equation}\label{F}
F(z_1,z_2,\dots,z_m)
=\dfrac{\partial^m}{\partial x_1\dots\partial x_m}
{\bf E}\left\{
D(z_1,\ldots,z_m;x_1,\ldots,x_m)\right\}\Bigg|_{x_1=\dots=x_m=0}.
\end{equation}
Here and below the symbol ${\bf E}\{\ldots\}$ denotes the
expectation with respect to the measure generated by $W$ (see (\ref{W})).

By using formulas (\ref{G_C}) and (\ref{G_Gr}), we obtain:
\[
\begin{array}{c}
D(z_1,\ldots,z_m;x_1,\ldots,x_m)
=\intd\exp\left\{-i\sum\limits_{\alpha=1}^m
\sum\limits_{j,k=1}^n\left(\dfrac{1}{\sqrt{n}}w_{j,k}+\delta_{j,k}(h_j^{(n)}-z_\alpha)
\right)\overline{\psi}_{j,\alpha}\psi_{k,\alpha}\right\}\\
\times\exp\left\{-i\sum\limits_{\alpha=1}^m\sum\limits_{j,k=1}^n\left(
\dfrac{1}{\sqrt{n}}w_{j,k}+\delta_{j,k}(h_j^{(n)}-z_\alpha-x_\alpha)
\right)\overline{\phi}_{j,\alpha}\phi_{k,\alpha}\right\}
\prod\limits_{j=1}^n d\,\Phi_j,
\end{array}
\]
where $\{\psi_{j,\alpha}\}_{j=1,\alpha=1}^{n \, m}$ are the
Grassmann variables ($n$ variables for each determinant in the
numerator), $\{\phi_{j,\alpha}\}_{j=1,\alpha=1}^{n \, m}$ are
complex ones ($n$ variables for each determinant in the
denominator),
$\Phi_j=(\phi_{j,1},\ldots,\phi_{j,m},\psi_{j,1},\ldots,\psi_{j,m})^t$
and
\[
d\,\Phi_j=\dfrac{1}{\pi^m}\prod\limits_{\alpha=1}^m \left(d\,\overline{\psi}_{j,\alpha}
d\,\psi_{j,\alpha}d\,\Re \phi_{j,\alpha}d\,\Im \phi_{j,\alpha}\right).
\]
Collecting separately the terms with $\Re w_{j,k}$ and $\Im w_{j,k}$
we get
\begin{multline}\label{Im_Re}
\intd\exp\left\{ i\sum\limits_{\alpha=1}^m
\sum\limits_{j=1}^n(z_\alpha-h_j^{(n)})\overline{\psi}_{j,\alpha}\psi_{j,\alpha}
+i\sum\limits_{\alpha=1}^m
\sum\limits_{j=1}^n(z_\alpha+x_\alpha-h_j^{(n)})\overline{\phi}_{j,\alpha}
\phi_{j,\alpha}\right\}\\
\times\exp\left\{-\dfrac{i}{\sqrt{n}}\sum\limits_{j<k}\Re w_{j,k}
\left({\Phi}^+_j\Phi_k+
\Phi^+_k\Phi_j\right)\right\}
\times\exp\left\{\dfrac{1}{\sqrt{n}}\sum\limits_{j<k}\Im w_{j,k}
\left({\Phi}^+_j\Phi_k-
\Phi^+_k\Phi_j\right)\right\}\\
\times\exp\left\{-\dfrac{i}{\sqrt{n}}\sum\limits_{j} w_{j,j}
\left(\Psi^+_j\Psi_j+\Phi^+_j\Phi_j\right)\right\}
\prod\limits_{j=1}^n d\,\Phi_j.
\end{multline}
Denote by $\exp\{f\}$ the first exponential. Integrating with respect
to the measure generated by $W$, we obtain after some calculations
\begin{multline}\label{E}
{\bf E}\left\{ D(z_1,\ldots,z_m;x_1,\ldots,x_m)\right\}\\
=\intd \exp\{f\}\cdot\exp\left\{ -\dfrac{1}{2n}\sum\limits_{j,k=1}^n\left(
\Phi^+_j\Phi_k\right)\left(\Phi^+_k
\Phi_j\right)\right\}\prod\limits_{j=1}^n
d\Phi_j.
\end{multline}
We will use below the following standart
\begin{lemma}[Hubbard-Stratonovitch transformation]\label{l:H-S}
We have in the above notations:
\begin{multline}\label{HS}
\exp\left\{ -\dfrac{1}{2n}\sum\limits_{j,k=1}^n\left(
\Phi^+_j\Phi_k\right)\left(\Phi^+_k\Phi_j\right)\right\}
=\intd\exp\left\{ -\dfrac{n}{2}\hbox{str}\, Q^2\right\}
\prod\limits_{j=1}^n\exp\{-i\Phi_j^+Q\Phi_j\}d\,Q,
\end{multline}
where
\[
Q=\left(\begin{array}{cc}
a&\sigma\\
\sigma^+&ib\\
\end{array}\right),
\]
$a=\{a_{j,k}\}_{j,k=1}^m,\,b=\{b_{j,k}\}_{j,k=1}^m$ are $m\times m$
Hermitian ordinary matrices, $\sigma=\{\sigma_{j,k}\}_{j,k=1}^m$
is a $m\times m$ matrix consisting of Grassmann variables
($\sigma^+$ is its Hermitian conjugate), and
\[
d\,Q=\dfrac{1}{\pi^{m^2}}\prod\limits_{j=1}^md\,a_{j,j}d\,b_{j,j}\prod\limits_{j<k}d\,\Re a_{j,k}
d\,\Im a_{j,k}d\,\Re b_{j,k}d\,\Im
b_{j,k}\prod\limits_{j,k=1}^md\,\overline{\sigma}_{j,k}d\,\sigma_{j,k}.
\]
\end{lemma}

\begin{proof}Define
\[
s^{(\psi)}_{\alpha,\beta}=\sum\limits_{j=1}^n\overline{\psi}_{j,\alpha}\psi_{j,\beta},
\quad s^{(\phi)}_{\alpha,\beta}=\sum\limits_{j=1}^n\overline{\phi}_{j,\alpha}
\phi_{j,\beta},\quad
s^{(\psi,\phi)}_{\alpha,\beta}=\sum\limits_{j=1}^n\overline{\psi}_{j,\alpha}
\phi_{j,\beta},\quad
s^{(\phi,\psi)}_{\alpha,\beta}=\sum\limits_{j=1}^n\psi_{j,\alpha}
\overline{\phi}_{j,\beta}.
\]
Write the sum at the exponent as:
\begin{multline}
-\dfrac{1}{2n}\sum\limits_{j,k=1}^n\left(
\Phi^+_j\Phi_k\right)\left(\Phi^+_k\Phi_j\right)
=
-\dfrac{1}{2n}\sum\limits_{\alpha,\beta=1}^m
s^{(\psi,\phi)}_{\alpha,\beta}\cdot
s^{(\phi,\psi)}_{\alpha,\beta}\\
-\dfrac{1}{2n}\sum\limits_{\alpha,\beta=1}^m
s^{(\psi,\phi)}_{\beta,\alpha}\cdot
s^{(\phi,\psi)}_{\beta,\alpha}
-\dfrac{1}{2n}\sum\limits_{\alpha,\beta=1}^ms^{(\psi)}_{\alpha,\beta}\cdot
s^{(\psi)}_{\beta,\alpha}-\dfrac{1}{2n}\sum\limits_{\alpha,\beta=1}^m
s^{(\phi)}_{\alpha,\beta}\cdot s^{(\phi)}_{\beta,\alpha}
\end{multline}
Now, use (\ref{G_C}) to obtain:
\[
\begin{array}{c}
\intd\prod\limits_{\alpha<\beta}
\dfrac{d\,\Im b_{\alpha,\beta} d\,\Re b_{\alpha,\beta}}{\pi}
\prod\limits_{\alpha}
\dfrac{d\, b_{\alpha,\alpha}}{\pi}
\exp\left\{\sum\limits_{\alpha,\beta=1}^m b_{\alpha\beta}
s^{(\psi)}_{\alpha,\beta}-n\sum\limits_{\alpha < \beta}\overline{b}_{\alpha,\beta}
b_{\alpha,\beta}-\dfrac{n}{2}\sum\limits_{\alpha=1}^m
b^2_{\alpha,\alpha} \right\}\\
=\intd\prod\limits_{\alpha<\beta}
\dfrac{d\,\Im b_{\alpha,\beta} d\,\Re b_{\alpha,\beta}}{\pi}
\exp\left\{\sum\limits_{\alpha<\beta}^m\Re b_{\alpha\beta}
\left(s^{(\psi)}_{\alpha,\beta}+s^{(\psi)}_{\beta,\alpha}\right)
\right\}\\
\times\exp\left\{i\sum\limits_{\alpha<\beta}^m\Im b_{\alpha\beta}
\left(s^{(\psi)}_{\alpha,\beta}-s^{(\psi)}_{\beta,\alpha}\right)
-n\sum\limits_{\alpha < \beta}\overline{b}_{\alpha,\beta}
b_{\alpha,\beta}-\dfrac{n}{2}\sum\limits_{\alpha=1}^m
b^2_{\alpha,\alpha} \right\}\\
=\left(\sqrt{\dfrac{2}{n}}\right)^m
\left(\sqrt{\dfrac{1}{n}}\right)^{m(m-1)}
\exp\left\{ \dfrac{1}{2n}\sum\limits_{\alpha,\beta=1}^ms^{(\psi)}_{\alpha,\beta}
s^{(\psi)}_{\beta,\alpha}\right\}
\end{array}
\]
Similar argument yields the formulas
\[
\begin{array}{c}
\intd\prod\limits_{\alpha<\beta}
\dfrac{d\,\Im a_{\alpha,\beta} d\,\Re a_{\alpha,\beta}}{\pi}
\prod\limits_{\alpha}
\dfrac{d\, a_{\alpha,\alpha}}{\pi}
\exp\left\{-i\sum\limits_{\alpha,\beta=1}^ma_{\alpha\beta}
s^{(\phi)}_{\alpha,\beta}
-n\sum\limits_{\alpha < \beta}\overline{a}_{\alpha,\beta}
a_{\alpha,\beta}-\dfrac{n}{2}\sum\limits_{\alpha=1}^m
a^2_{\alpha,\alpha} \right\}\\
=\left(\sqrt{\dfrac{2}{n}}\right)^m
\left(\sqrt{\dfrac{1}{n}}\right)^{m(m-1)}
\exp\left\{ -\dfrac{1}{2n}\sum\limits_{\alpha,\beta=1}^m
s^{(\phi)}_{\alpha,\beta}s^{(\phi)}_{\beta,\alpha}\right\}
\end{array}
\]
and
\[
\begin{array}{c}
\intd\prod\limits_{\alpha,\beta=1}^m
d\,\eta_{\alpha,\beta} d\,\eta_{\alpha,\beta}
\exp\left\{\sum\limits_{\alpha,\beta=1}^m\eta_{\alpha\beta}
s^{(\psi,\phi)}_{\alpha,\beta}+
\sum\limits_{\alpha,\beta=1}^m\overline{\eta}_{\alpha\beta}
s^{(\phi,\psi)}_{\alpha,\beta}
-n\sum\limits_{\alpha, \beta=1}^m\overline{\eta}_{\alpha,\beta}
\eta_{\alpha,\beta}\right\}\\
=\prod\limits_{\alpha,\beta=1}^m\intd
d\,\eta_{\alpha,\beta} d\,\eta_{\alpha,\beta}
\left(1+\eta_{\alpha,\beta}s^{(\psi,\phi)}_{\alpha,\beta}+
\overline{\eta}_{\alpha,\beta}s^{(\phi,\psi)}_{\alpha,\beta}
-n \overline{\eta}_{\alpha,\beta}\eta_{\alpha,\beta}+
\overline{\eta}_{\alpha,\beta}\eta_{\alpha,\beta}
s^{(\psi,\phi)}_{\alpha,\beta}s^{(\phi,\psi)}_{\alpha,\beta}\right)\\
=n^{m^2} \prod\limits_{\alpha,\beta=1}^m\left(1-\dfrac{1}{n}
s^{(\psi,\phi)}_{\alpha,\beta}s^{(\phi,\psi)}_{\alpha,\beta}\right)=
n^{m^2}\exp\left\{-\dfrac{1}{n}
s^{(\psi,\phi)}_{\alpha,\beta}s^{(\phi,\psi)}_{\alpha,\beta}\right\},
\end{array}
\]
where we used (\ref{G_Gr}) to obtain the last formula.
Collecting together three above formulas, we present the
l.h.s. of (\ref{HS}) as
\[
\begin{array}{c}
\dfrac{1}{2^m}\intd\prod\limits_{\alpha,\beta=1}^m
d\,\eta_{\alpha,\beta} d\,\eta_{\alpha,\beta}
\prod\limits_{\alpha<\beta}
\dfrac{d\,\Im a_{\alpha,\beta} d\,\Re a_{\alpha,\beta}}{\pi}
\prod\limits_{\alpha}
\dfrac{d\, a_{\alpha,\alpha}}{\pi}
\prod\limits_{\alpha<\beta}
\dfrac{d\,\Im b_{\alpha,\beta} d\,\Re b_{\alpha,\beta}}{\pi}
\prod\limits_{\alpha}
\dfrac{d\, b_{\alpha,\alpha}}{\pi}\\
\times\exp\left\{\sum\limits_{\alpha,\beta=1}^ma_{\alpha\beta}
s^{(\psi)}_{\alpha,\beta}+i\sum\limits_{\alpha,\beta=1}^mb_{\alpha\beta}
s^{(\phi)}_{\alpha,\beta}+
\sum\limits_{\alpha,\beta=1}^m\eta_{\alpha\beta}
s^{(\psi,\phi)}_{\alpha,\beta}+
\sum\limits_{\alpha,\beta=1}^m\overline{\eta}_{\alpha\beta}
s^{(\phi,\psi)}_{\alpha,\beta}
\right\}\\
\times\exp\left\{-\dfrac{n}{2}\sum\limits_{\alpha=1}^m
(a_{\alpha\alpha}^2+b_{\alpha\alpha}^2)-n\sum\limits_{\alpha<\beta}
(\overline{a}_{\alpha\beta}a_{\alpha\beta}+\overline{b}_{\alpha\beta}
b_{\alpha\beta})-
n\sum\limits_{\alpha,\beta=1}^m
\overline{\eta}_{\alpha\beta}\eta_{\alpha\beta}\right\}\\
=\dfrac{1}{2^m}\intd\exp\left\{ -\dfrac{n}{2}\hbox{str}\, Q^2\right\}
\prod\limits_{j=1}^n\exp\{-i\Phi_j^+Q\Phi_j\}d\,Q,
\end{array}
\]
where the matrix $Q$ is described in the lemma.
\end{proof}

The above allows us to rewrite the integral in the
r.h.s. of (\ref{E}) as:

\begin{equation}\label{Q1}
\begin{array}{c}
\dfrac{1}{2^m}\intd\exp\{f\}\cdot\prod\limits_{j=1}^n
d\,\Phi_j\cdot\exp\left\{ -\dfrac{n}{2}\hbox{str}\, Q^2\right\}
\prod\limits_{j=1}^n\exp\{-i\Phi_j^+Q\Phi_j\}d\,Q.
\end{array}
\end{equation}
Setting
\[
\Lambda=
\left(\begin{array}{cccccccc}
z_1&0&\ldots&0&0&0&\ldots&0\\
0&z_2&\ldots&0&0&0&\ldots&0\\
\ldots&\ldots&\ldots&\ldots&\ldots&\ldots&\ldots&\ldots\\
0&\ldots&0&z_n&0&\ldots&0&0\\
0&0&\ldots&0&z_1+x_1&0&\ldots&0\\
0&\ldots&0&0&0&z_2+x_2&\ldots&0\\
\ldots&\ldots&\ldots&\ldots&\ldots&\ldots&\ldots&\ldots\\
0&\ldots&0&0&0&\ldots&0&z_n+x_n\\
\end{array}\right),
\]
and using the explicit form of $\exp\{f\}$,  we obtain from
(\ref{Q1})
\begin{equation}\label{Q2}
\begin{array}{c}
\dfrac{1}{2^m}\intd\prod\limits_{j=1}^n
d\,\Phi_j\cdot\exp\left\{ -\dfrac{n}{2}\hbox{str} Q^2\right\}
\prod\limits_{j=1}^n\exp\{-i\Phi_j^+(Q-\Lambda+h_jI)\Phi_j\}d\,Q.
\end{array}
\end{equation}
Recall now that $\Im z_1=\ldots=\Im z_m=-\varepsilon<0$. Hence,
$\Lambda=\Lambda_1 -\varepsilon\, I$, where $\Lambda_1$ is
a matrix, whose entries are the real parts of the entries of $\Lambda$.

We integrate (\ref{Q2}) with respect to $\psi$ and $\phi$ by using (\ref{G_comb}),
as a result the integral (\ref{E}) is equal to
\begin{multline}
\dfrac{1}{2^m}\intd\exp\left\{ -\dfrac{n}{2}\hbox{str}\, Q^2\right\}
\prod\limits_{j=1}^n\hbox{sdet}\,(Q-\Lambda+h_jI)^{-1}d\,Q\\
=\dfrac{1}{2^m}\intd\exp\left\{ -\dfrac{n}{2}\hbox{str}\, Q^2\right\}
\prod\limits_{j=1}^n\hbox{sdet}\,(Q-\Lambda_1+\varepsilon\cdot I+h_jI)^{-1}d\,Q\\
=\dfrac{1}{2^m}\intd\exp\left\{ -\dfrac{n}{2}\hbox{str} (Q+\Lambda_1)^2\right\}
\prod\limits_{j=1}^n\hbox{sdet}(Q+\varepsilon\cdot I+h_jI)^{-1}d\,Q
\end{multline}

Write $Q=U^{-1}SU$, where $U$ is a unitary
super-matrix and the matrix $S$ is
\[
S=\left(\begin{array}{cc}
S_1&0\\
0&S_2\\
\end{array}\right),
\]
where
\[
S_1=\left(\begin{array}{ccccc}
s_1&0&\ldots&0&0\\
0&s_2&0&\dots&0\\
\ldots&\ldots&\ldots&\ldots&\ldots\\
0&0&\ldots&0&s_n\\
\end{array}\right),\quad
S_2=\left(\begin{array}{ccccc}
it_1&0&\ldots&0&0\\
0&it_2&0&\dots&0\\
\ldots&\ldots&\ldots&\ldots&\ldots\\
0&0&\ldots&0&it_n\\
\end{array}\right).
\]
Use the super-analog (\ref{Itz_Z_gr}) of the Itzykson-Zuber
formula for the integration over the unitary group (see
\cite{Gu:91}). This yields
\begin{multline}\label{E_ok}
{\bf E}\left\{ D(z_1,\ldots,z_m;x_1,\ldots,x_m)\right\}=1-\chi(x_1,\ldots,x_m)\\
+\dfrac{(-2\pi)^{-m}n^m}{B_m(\Lambda)}\intd
\exp\left\{ -\dfrac{n}{2}\hbox{str} (S+\Lambda_1)^2\right\}
\prod\limits_{j=1}^n\prod\limits_{\alpha=1}^m\left(\dfrac{it_\alpha+i\,\varepsilon+h_j^{(n)}}
{s_\alpha+i\,\varepsilon+h_j^{(n)}}\right)\cdot B_m(S)\prod\limits_{\alpha=1}^md\,t_\alpha
d\,s_\alpha,
\end{multline}
where $B_m(S)$ is the Cauchy determinant (\ref{B}).

  Using the formula for the Cauchy determinant, we obtain that
\[
B_m(\Lambda)^{-1}=\prod\limits_{\alpha=1}^m x_\alpha
\prod\limits_{\alpha>\beta}\dfrac{(z_\alpha-z_\beta)(z_\alpha+x_\alpha-
z_\beta-x_\beta)}{(z_\alpha-z_\beta-x_\beta)(z_\beta-z_\alpha-x_\alpha)}.
\]
Substituting this to (\ref{E_ok}), differentiating (\ref{E_ok}) with respect
to every $x_\alpha$ and
putting then $x_1=\ldots=x_m=0$, we have
\begin{equation}\label{E_pr}
\begin{array}{c}
\dfrac{\partial^m}{\partial x_1\dots\partial x_m}
{\bf E}\left\{ D(z_1,\ldots,z_m;x_1,\ldots,x_m)\right\}\Bigg|_{x_1=\dots=x_m=0}\\
=\dfrac{n^m}{(-2\pi)^m}\intd
\exp\left\{ -\dfrac{n}{2}\hbox{str} (S+\widetilde{\Lambda})^2\right\}
\prod\limits_{j=1}^n\prod\limits_{\alpha=1}^m\left(\dfrac{it^+_\alpha+h_j^{(n)}}
{s^+_\alpha+h_j^{(n)}}\right)\cdot B_m(S)\prod\limits_{\alpha=1}^md\,t_\alpha
d\,s_\alpha,
\end{array}
\end{equation}
where $\widetilde{\Lambda}=\Lambda_1\big|_{x_1=\dots=x_m=0}$,
$s^+_\alpha=s_\alpha+i\,\varepsilon$,
$it^+_\alpha=it_\alpha+i\,\varepsilon$.

 We can write the determinant (\ref{B}) as
\[
B_m(S)=\sum\limits_\omega(-1)^{\sigma(\omega)}\prod\limits_{\alpha=1}^m
\dfrac{1}{s_\alpha-it_{\omega(\alpha)}},
\]
where the sum is over all permutations $\omega$ of the indices
$\{1,\ldots,m\}$, and $\sigma(\omega)$ is the parity of a permutation.
The rest of the integrand factorizes in a $m$-fold product. Hence,
recalling the definition of $F$ in (\ref{F}), we obtain finally
\begin{equation}\label{Det1}
F(z_1,\ldots,z_m)=\det\{\widehat{K}_n(z_\alpha,z_\beta)\},
\end{equation}
where
\begin{multline}\label{K_hat}
\widehat{K}_n(z_\alpha,z_\beta)\\
=-\dfrac{n}{2\pi}\intd
\exp\left\{ -\dfrac{n}{2}((s_\alpha+\Re z_\alpha)^2-(it_\beta+\Re z_\beta)^2)\right\}
\prod\limits_{j=1}^n\left(\dfrac{it^+_\beta+h_j^{(n)}}
{s_\alpha^++h_j^{(n)}}\right)\dfrac{d\,t_\beta
d\,s_\alpha}{s_\alpha-it_\beta}.
\end{multline}
Denote
\begin{multline}\label{K_tild}
\widetilde{K}_n(\lambda,\mu)=\\
=-\dfrac{n}{2\pi^2}\intd\dfrac{d\,t
d\,s}{s-it}
\exp\left\{ -\dfrac{n}{2}((s+\lambda)^2-(it+\mu)^2)\right\}
\lim\limits_{\varepsilon\to 0}\Im\prod\limits_{j=1}^n
\dfrac{it^++h_j^{(n)}}{s^++h_j^{(n)}}\\
=-\dfrac{n}{2\pi^2}\intd\dfrac{d\,t
d\,s}{s-it}
\exp\left\{ -\dfrac{n}{2}((s+\lambda)^2-(it+\mu)^2)\right\}
\prod\limits_{j=1}^n(it+h_j^{(n)})
\lim\limits_{\varepsilon\to 0}\Im\prod\limits_{j=1}^n
\dfrac{1}{s^++h_j^{(n)}}.
\end{multline}
Changing variables to $it\to\, -t$, $s\to -s$, we obtain
\begin{multline}\label{K1}
\widetilde{K}_n(\lambda,\mu)\\
=-\dfrac{in}{2\pi^2}\intd\dfrac{d\,t
d\,s}{s-t}
\exp\left\{ -\dfrac{n}{2}((-s+\lambda)^2-(-t+\mu)^2)\right\}
\prod\limits_{j=1}^n(t-h_j^{(n)})
\lim\limits_{\varepsilon\to 0}\Im\prod\limits_{j=1}^n
\dfrac{1}{s^--h_j^{(n)}},
\end{multline}
where $s^-=s-i\varepsilon$.

Note that we can assume without loss of generality that
$\{h_j^{(n)}\}_{j=1}^n$ are distinct and then we have on the sense
of distributions
\[
\lim\limits_{\varepsilon\to 0}\Im\prod\limits_{j=1}^n\dfrac{1}{s^--h_j^{(n)}}=
\pi\sum\limits_{j=1}^n
\delta(s-h_j^{(n)})\prod\limits_{k\ne j}\dfrac{1}{h_j^{(n)}-h_k^{(n)}}.
\]
Hence, the integral in the r.h.s. of (\ref{K1}) is
\begin{equation}\label{int_t}
-\dfrac{in}{2\pi}\intd d\,t\sum\limits_{j=1}^n
\dfrac{\exp\{-\dfrac{n}{2}((-h_j^{(n)}+\lambda)^2-(-t+\mu)^2)\}}{h_j^{(n)}-t}
\prod\limits_{k=1}^n\left(t-h_k\right)\prod\limits_{k\ne
j}\dfrac{1}{h_j^{(n)}-h_k^{(n)}},
\end{equation}
where the integration with respect to $t$ is taken over the imaginary axis.

We will replace now the integral with respect to $t$ to that over $L$
parallel to the imaginary axis and lying to the left of all
$\{h_j^{(n)}\}_{j=1}^n$. To do this we consider the rectangle whose vertical
sides lie on the imaginary axis and on $L$, and the horizontal
ones lie
on the lines $\Re z=\pm R$. The integral (\ref{int_t}) over
this contour is zero, since there are no singularities inside the
contour. The integrals over the horizontal segments of the contour tends to
zero, as $R\to \infty$, because of the term $-t^2/2$ in the
exponent of (\ref{int_t}). Therefore, the
integrals (\ref{int_t}) over the imaginary axis and $L$ are
equal.

Now using the residue theorem for the contour over $v$, we can get
that
\begin{multline}\label{K2}
\widetilde{K}_n(\lambda,\mu)\\
=-n\intd\limits_{L}\dfrac{d\,t}{2\pi}\oint\limits_{C}\dfrac{d\,v}{2\pi}
\dfrac{\exp\left\{ -\dfrac{n}{2}((-v+\lambda)^2-(-t+\mu)^2)\right\}
}{v-t}\prod\limits_{j=1}^n\left(\dfrac{t-h_j^{(n)}}
{v-h_j^{(n)}}\right),
\end{multline}
where the contour $C$ has all $\{h_j^{(n)}\}_{j=1}^n$ inside and does not
intersect $L$.

To finish the proof of Proposition we need
\begin{lemma}\label{l:2}
Let $\{R_n^{(m)}\}_{m\ge 1}$ be defined in (\ref{R}),
$\Im z_1=\ldots=\Im z_m=-\varepsilon <0$ and $\Re
z_j=\lambda_j$ are distinct. Then
\[
R^{(n)}_m(\lambda_1,\ldots,\lambda_m)=\lim\limits_{\varepsilon\to
0} \dfrac{1}{\pi^m}{\bf
E}\left\{\prod\limits_{k=1}^m\Im\,\hbox{Tr}\,\dfrac{1}{H_n-z_k}
\right\}.
\]
\end{lemma}
\begin{proof} Let $\{\lambda_i^{(n)}\}_{i=1}^n$ be the eigenvalues of the
matrix $H_n$. To make the proof more clear, let us consider the cases $m=1$
and $m=2$.

1)~$m=1$. Putting in (\ref{R}) $\varphi_1(\lambda)=
\Im \dfrac{1}{\lambda-z}$ we have
\begin{equation}\label{R1}
{\bf E}\left\{\Im\,\hbox{Tr}\,\dfrac{1}{H_n-z}\right\}=\sum\limits_{j=1}^n
{\bf E}\left\{\Im\dfrac{1}{\lambda_j^{(n)}-z}\right\}=\Im\intd\dfrac{R_1^{(n)}(d\,\mu)}
{\mu-z}.
\end{equation}
It was proved before that the l.h.s. of (\ref{R1}) has a limit, as $\varepsilon\to
0$ (see (\ref{F1}),(\ref{Det1}), (\ref{K_tild}) and (\ref{K2})). Therefore
the r.h.s. of (\ref{R1}) has a
limit too. Hence, according Stieltjes-Perron formula, the measure
$R_1^{(n)}(d\,\mu)$ has a density $R_1^{(n)}(\mu)$ and this
density is equal to the limit of the l.h.s. of (\ref{R1}) , i.e., $\widetilde{K}_n(\mu,\mu)$.
Since $\widetilde{K}_n$ is defined by the integral (\ref{K2}), $R_1^{(n)}$ is
bounded.

2)~$m=2$. Putting in (\ref{R}) $\varphi_1(\lambda)=\Im
\dfrac{1}{\lambda-z_1} \Im
\dfrac{1}{\lambda-z_2}$,
$\varphi_2(\lambda_1,\lambda_2)=\Im \dfrac{1}{\lambda_1-z_1}
\Im \dfrac{1}{\lambda_2-z_2}$
we have
\begin{equation}\label{R2}
\begin{array}{c}
{\bf E}\left\{\Im\,\hbox{Tr}\,\dfrac{1}{H_n-z_1}
\Im\,\hbox{Tr}\,\dfrac{1}{H_n-z_2}\right\}
=\displaystyle{\sum\limits_{j=1}^n}\,{\bf E}\left\{\Im\dfrac{1}{\lambda_j^{(n)}-z_1}
\Im\dfrac{1}{\lambda_j^{(n)}-z_2}\right\}\\
+\displaystyle{\sum\limits_{j_1\ne j_2}}\,{\bf E}\left\{\Im\dfrac{1}{\lambda_{j_1}^{(n)}-z_1}
\Im\dfrac{1}{\lambda_{j_2}^{(n)}-z_2}\right\}=
\intd R_1^{(n)}(\mu)\Im\left(\dfrac{1}{\mu-z_1}\right)
\Im\left(\dfrac{1}{\mu-z_2}\right)d\,\mu\\+\intd R_2^{(n)}(d\,\mu_1, d\,\mu_2)
\Im\left(\dfrac{1}{\mu_1-z_1}\right)\Im\left(\dfrac{1}{\mu_2-z_2}\right).
\end{array}
\end{equation}
Consider the limit of the integral
\[
I_1=\intd R_1^{(n)}(\mu)\Im\left(\dfrac{1}{\mu-z_1}\right)
\Im\left(\dfrac{1}{\mu-z_2}\right)d\,\mu
\]
where $\Im z_1=\Im z_2=-\varepsilon$, $\Re z_1=\lambda_1$, $\Re
z_2=\lambda_2$ and $\lambda_1\ne \lambda_2$, as $\varepsilon\to 0$.
It is easy to see that
\[
I_1=\intd\dfrac{\varepsilon^2R_1^{(n)}(\mu)d\,\mu}{((\lambda_1-\mu)^2+
\varepsilon^2)((\lambda_2-\mu)^2+ \varepsilon^2)}.
\]
Let us make the change $\varepsilon\nu=\lambda_1-\mu$. We obtain
\[
I_1=-\intd\dfrac{\varepsilon
R_1^{(n)}(\lambda_1-\varepsilon\nu)d\,\nu}
{(\nu^2+1)((\lambda_2-\lambda_1+\varepsilon\nu)^2+
\varepsilon^2)}.
\]
$R_1^{(n)}(\lambda_1-\varepsilon\nu)$ is bounded (as it was proved
before), and so, $$\lim\limits_{\varepsilon\to 0}I_1=0.$$ Now
consider the integral
\[
I_2=\intd R_2^{(n)}(d\mu_1,
d\mu_2)\Im\left(\dfrac{1}{\mu_1-z_1}\right)
\Im\left(\dfrac{1}{\mu_2-z_2}\right).
\]
Since we proved that $\lim\limits_{\varepsilon\to 0}I_1=0$, the
limit of $I_2$, as $\varepsilon\to 0$, is equal to the limit of the
l.h.s. of (\ref{R2}) (which exists according to (\ref{F1}),(\ref{Det1}), (\ref{K_tild})
and (\ref{K2})). Again by the Stieltjes-Perron formula
the measure $R_2^{(n)}(d\,\mu_1, d\,\mu_2)$ has a density
$R_2^{(n)}(\mu_1,\mu_2)$, and this density is equal to the limit
of the l.h.s, i.e., $\det\{\widetilde{K}_n(\mu_i,\mu_j)\}_{i,j=1}^2$. Since
$\widetilde{K}_n$ is defined by the integral (\ref{K2}), $R_2^{(n)}$ is bounded.

Therefore, $$ \lim\limits_{\varepsilon\to 0} \dfrac{1}{\pi^2}{\bf
E}\left\{\Im\,\hbox{Tr}\,\dfrac{1}{H_n-z_1}
\Im\,\hbox{Tr}\,\dfrac{1}{H_n-z_2} \right\}=\pi^2
R^{(n)}_2(\lambda_1,\lambda_2).$$ For $m>2$ the proof is similar
(we should use that $R_1^{(n)},\ldots,R_{m-1}^{(n)}$ are bounded).
\end{proof}

Now (\ref{Det1}), (\ref{K_hat}), and Lemma \ref{l:2} yield
formula (\ref{Det}) for the correlation function (\ref{R}),
where $\widetilde{K}_n$ is defined by (\ref{K2}). The multiplier
$\exp\{\mu^2-\lambda^2\}$ vanishes during the calculation of the
determinant, and so we can omit it. Finally we have
formula (\ref{K}).
\end{proof}

\section{Proof of the Theorem \protect\ref{thm:1}.}

In this section we will prove Theorem \ref{thm:1}, using
(\ref{Det}) and making the limiting transition in (\ref{K}).
Putting in formula (\ref{K}) $ \lambda=\lambda_0+\lambda^\prime/n$
and $\mu=\lambda_0+\mu^\prime/n$, we get:
\begin{equation}  \label{Ker}
K_n(\lambda,\mu)=-n\displaystyle\int\limits_L\displaystyle\frac{dt}{2\pi}
\oint\limits_{C}\displaystyle\frac{dv}{2\pi}
\exp\{v\lambda^\prime-t\mu^
\prime\}\displaystyle\frac{\exp\{n(S_n(t,\lambda_0)-
S_n(v,\lambda_0))\}}{v-t },
\end{equation}
where
\begin{equation}  \label{S_n}
S_n(z,\lambda_0)=\displaystyle\frac{z^2}{2}+\displaystyle\frac{1}{n}
\sum\limits_{i=1}^{n}\ln(z-h_j^{(n)})- \lambda_0z,
\end{equation}
and $C$ is an arbitrary contour having all $\{h_j^{(n)}\}_{j=1}^n$ inside, $L
$ is a line parallel to the imaginary axis and lying to the left of $C$.
Formula (\ref{Det}) reduces (\ref{Un}) to the proof of the following
relation:
\begin{equation*}
\lim\limits_{n\to\infty}\displaystyle\frac{1}{n\rho_n(\lambda_0)}
K_n(\lambda,\mu)= S(\lambda^\prime -\mu^\prime),
\end{equation*}
where $S$ is defined in (\ref{S}).

We will choose now the contour $C$ as follows. Define
\begin{equation}  \label{g_0,n}
g^{(0)}_n(z)=\displaystyle\frac{1}{n}\displaystyle\sum\limits_{j=1}^n
\displaystyle\frac{1}{h_j^{(n)}-z},
\end{equation}
and consider the equation
\begin{equation}  \label{z}
z-g^{(0)}_n(z)=\lambda.
\end{equation}
Equation (\ref{z}) can be written as a polynomial equation of
degree $(n+1)$ and so it has $(n+1)$ roots. Considering the real
$z$ and taking into account that if $z\to h_j^{(n)}+0$, then the
l.h.s. tends to $+\infty$, and if $z\to h_j^{(n)}-0$ then the
l.h.s. tends to $-\infty$, we have that $n-1$ of these roots are
always real and belong to the segments between $h_j^{(n)}$ . If
$\lambda$ is big enough, then all $n+1$ roots are real. Let $
z_n(\lambda)$ be a root which has the order
$\lambda-1/\lambda+O(1/\lambda^2) $, as $\lambda\to\infty$. If
$\lambda$ decreases, then $z_n(\lambda)$ will decrease too, and
coming to some $\lambda_{c_1}$ the real root disappears and there
appear two complex ones -- $z_n(\lambda)$ and $\overline{
z_n(\lambda)}$. Then $z_n(\lambda)$ may be real again, than again
complex, and so on, however as soon as $\lambda$ becomes less then
some $\lambda_{c_2} $, the root becomes again real. Choose $C_n$
to be the union of two curves -- $z_n(\lambda)$ and
$\overline{z_n(\lambda)}$, corresponding to $\lambda$ such that
$\Im z_n(\lambda)\ne 0$. It is clear that the set of such
$\lambda$ is $\bigcup\limits_{j=1}^k I_k$, where $\{I_j\}_{j=1}^k$
are non intersecting segments. It is easy to see also that the
contour $C_n$ is closed and has all $\{h_j^{(n)}\}_{j=1}^n$
inside.

Let us consider the limiting equation
\begin{equation}  \label{eqv_g_0}
z-g^{(0)}(z)=\lambda, \hbox{where}\quad
g^{(0)}(z)=\displaystyle\int
\displaystyle\frac{N^{(0)}(d\lambda)}{\lambda-z},
\end{equation}
where $\lambda\in \mathbb{R}$ is fixed. We have

\begin{lemma}
\label{l:3} Equation (\ref{eqv_g_0}) has a unique solution in the upper
half-plane $\Im z>0$, if $\lambda\in \emph{supp}\, N$, and has no solutions,
if $\lambda\not\in \emph{supp}\, N$. The solution is continuous in $\lambda$
in the domain where it exists.
\end{lemma}

\begin{proof}
Set
 \[
g(z)=\intd\dfrac{N(d\lambda)}{\lambda-z}.
\]
Then we have \cite{Pa:72}
\begin{equation}\label{eqv_g}
g(z)=g^{(0)}(z+g(z)).
\end{equation}
Note, that the measure $N$ is absolutely continuous.
Indeed, it follows from (\ref{eqv_g}) that
\[
|\Im g|\le\displaystyle{\int} \dfrac{|\Im z+\Im g|N^{(0)}(d\,\lambda)}
{(\lambda-\Re z-\Re g)^2+(\Im z+\Im g)^2}\le \dfrac{1}{|\Im g+\Im z|}
=\dfrac{1}{|\Im g|+|\Im z|},
\]
thus
\[
|\Im g|\le 1.
\]
According to the standard theory, it means that there exists
$\lim\limits_{\varepsilon\to +0} \Im g(\lambda+i\varepsilon)$ and
so the measure $N$ is absolutely continuous.

 Put $z(\lambda)=\lambda+i\,0+g(\lambda+i\,0)$, if $\lambda\in
\hbox{supp}\, N$. Using (\ref{eqv_g}), we
obtain that
\[
z(\lambda)-g^{(0)}(z(\lambda))=\lambda.
\]
Hence, the solution exists if $\lambda\in \hbox{supp}\, N$.
It is easy to see that the contour $C_\infty$ constructed by the
union of the curves $z(\lambda)$ and $\overline{z(\lambda)}$,
intersects the real axis at the points where $1-\dfrac{d}{d\,x}g^{(0)}(x)\ge 0$.

 Let us prove the uniqueness of the solution. Let
$z=x+iy$ be a solution of (\ref{eqv_g_0}), $y>0$. Then,
considering the imaginary part (\ref{eqv_g_0}), we obtain
\begin{equation}\label{*}
\intd\dfrac{N^{(0)}(d\lambda)}{(x-\lambda)^2+y^2}=1.
\end{equation}
If $x$ is real and outside $C_\infty$, then
$1-\dfrac{d}{d\,x}g^{(0)}(x)>0$, hence
\[
1-\intd\dfrac{N^{(0)}(d\lambda)}{(x-\lambda)^2}>0,
\]
thus
\[
\intd\dfrac{N^{(0)}(d\lambda)}{(x-\lambda)^2+y^2}<\intd\dfrac{N^{(0)}(d\lambda)}
{(x-\lambda)^2}<1,
\]
and there are no solutions. If $x$ is inside $C_\infty$,
then the solution with respect to $y$ is unique (since the
r.h.s of (\ref{*}) is monotone in $y$) and this solution is
found already, it is $z(\lambda)$. For this solution $z-g^{(0)}(z)$
belongs to $\hbox{supp}\, N$. So, we are left to
prove the continuity of $z(\lambda)$. Let $\lambda_0\in\hbox{supp}\, N$. Consider
$F(z)=z-g^{(0)}(z)-\lambda_0$ and $f_\lambda(z)=\lambda_0-\lambda$.
It was proved before that $F(z)$ has a unique
root $z(\lambda_0)$ in the upper half-plane. Denote
$\omega=\{z:|z-z(\lambda_0)|=\varepsilon\}$. There exists
$\delta>0$ such that $|F(z)|>\delta$. Therefore, if $\lambda\in
U_\delta(\lambda)$ and $z\in \omega$ we have
\[
|F(z)|>\delta>|f(z)|.
\]
It follows from the Rouchet theorem that for any $\lambda\in
U_\delta(\lambda)$ the function $F(z)+f(z)=z-g^{(0)}(z)-\lambda$ has
the same number of roots as $F(z)$ inside $\omega$, i.e., one. This
proves the continuity of $z(\lambda)$. The lemma is proved.
\end{proof}

Let us study the behavior of the function $\Re S_n(z_n(\lambda),\lambda_0)$
of (\ref{S_n}) on the contour $C_n$.

\begin{lemma}
\label{l:4}  Let $z$ belong to the upper part of $C_n$, i.e., $
z=z_n(\lambda)= x_n(\lambda)+iy_n(\lambda)$, $y_n(\lambda)>0$,
$\lambda\in \bigcup\limits_{j=1}^k I_j$, where
\begin{equation}  \label{zn}
z_n(\lambda)-g^{(0)}_n(z_n(\lambda))= \lambda.
\end{equation}
Then $\Re\,S_n(z_n(\lambda),\lambda_0)\ge 0$, and the equality holds only at
$\lambda=\lambda_0$.
\end{lemma}

\begin{proof}
The real and the imaginary parts of (\ref{zn}) yield for $x_n=\Re z_n$ and
$y_n=\Im z_n$:
\begin{equation}
\label{xyn}\left\{
\begin{array}{rl}
x_n(\lambda)+\dfrac{1}{n}\sumd\limits_{j=1}^n\dfrac{x_n(\lambda)-h_j^{(n)}}{(x_n(\lambda)-h_j^{(n)})^2+
y^2_n(\lambda)}&= \lambda,\\
y_n(\lambda)\left(1-\dfrac{1}{n}\sumd\limits_{j=1}^n\dfrac{1}{(x_n(\lambda)-h_j^{(n)})^2+
y^2_n(\lambda)}\right)&= 0,
\end{array}\right.
\end{equation}
Differentiate (\ref{zn}) with respect to $\lambda$:
\[
z_n^\prime(\lambda)\left(1-\dfrac{d}{d\,z}g^{(0)}_n(z_n(\lambda)) \right)=1,
\,\,\,\hbox{i.e.,}\,
\]
\begin{equation}\label{zn_pr}
z_n^\prime(\lambda)=\left(1-\dfrac{d}{d\,z}g^{(0)}_n(z_n(\lambda))\right)^{-1},
\end{equation}
where $g^{(0)}_n(z)$ is defined in (\ref{g_0,n}).

It follows from the implicit function theorem that $C_n$ intersects the
real axis at the points where
\[
1-\dfrac{d}{d\,x}g^{(0)}_n(x)=0.
\]
Since
\[
\dfrac{d}{d\,x}g^{(0)}_n(x)=\dfrac{1}{n}\sum\limits_{j=1}^n\dfrac{1}{(x-h_j^{(n)})^2},
\]
the inequality $1-\dfrac{d}{d\,x}g^{(0)}_n(x)<0$ holds near $h_j^{(n)}$.
Thus, the function $1-\dfrac{d}{d\,x}g^{(0)}_n(x)$
is always positive outside $C_n$. On the other hand, $z_n(\lambda)=x_n(\lambda)$
outside $C_n$ and in this case
\[
x^\prime_n(\lambda)=z^\prime_n(\lambda)=\left(1-
\dfrac{d}{d\,z}g^{(0)}_n(z_n(\lambda))\right)^{-1}>0.
\]
Now let $\lambda\in\bigcup\limits_{j=1}^k I_j$, i.e.,
$z_n(\lambda)$ belongs to $C_n$. We get from (\ref{zn_pr})
\[
\Re z_n^\prime(\lambda)=x_n^\prime(\lambda)=\Re\left(\left(1-
\dfrac{d}{d\,z}g^{(0)}_n(z_n(\lambda))\right)^{-1}\right)
=\dfrac{a_n(\lambda)}{a^2_n(\lambda)+b^2_n(\lambda)},
\]
where
\begin{equation}
\label{ab}\left\{
\begin{array}{rl}
a_n(\lambda)&= \Re\left(1-\dfrac{d}{d\,z}g^{(0)}_n(z_n(\lambda))\right) ,\\
b_n(\lambda)&=
\Im\left(1-\dfrac{d}{d\,z}g^{(0)}_n(z_n(\lambda))\right),
\end{array}\right.
\end{equation}
and hence
\[
a_n(\lambda)= 1-\dfrac{1}{n}\sumd\limits_{j=1}^n
\dfrac{(x_n(\lambda)-h_j^{(n)})^2-y^2_n(\lambda)}{((x_n(\lambda)-h_j^{(n)})^2+y^2_n(\lambda))^2}.
\]
Taking into account that $y_n(\lambda)\not= 0$, we obtain from (\ref{xyn}) that
\begin{equation}\label{cond}
1=\dfrac{1}{n}\sumd\limits_{j=1}^n\dfrac{1}{(x_n(\lambda)-h_j^{(n)})^2+y^2_n(\lambda)}.
\end{equation}
This and the previous equation yield
\begin{equation}
\label{a}
a_n(\lambda)=\dfrac{1}{n}\sumd\limits_{j=1}^n\dfrac{2y^2_n(\lambda)}
{((x_n(\lambda)-h_j^{(n)})^2+y^2_n(\lambda))^2}>0.
\end{equation}
 It follows from (\ref{ab}) and (\ref{a}) that in this case $x_n^\prime(\lambda)>0$
too (if only $y_n(\lambda)\not=0$). Hence, $x_n(\lambda)$ is a
monotone increasing function defined everywhere in $\mathbb{R}$.

  Consider $\Re S_n(z,\lambda_0)$ on
the upper part of $C_n$. Substituting the
expression  $z_n(\lambda)= x_n(\lambda)+iy_n(\lambda)$ into (\ref{S_n}),
$y_n(\lambda)>0$, we obtain
\[
\Re S_n(z_n(\lambda),\lambda_0)=\dfrac{x^2_n(\lambda)-y^2_n(\lambda)}{2}+\dfrac{1}{n}
\Re\sumd\limits_{j=1}^n \ln(x_n(\lambda)+iy_n(\lambda)-h_j^{(n)})-\lambda_0 x_n(\lambda)+C.
\]
Differentiating this equality and using (\ref{cond}), we
get
\begin{equation}
\label{S1}
\Re S_n(z_n(\lambda),\lambda_0)^\prime=
x_n^\prime(\lambda)(\lambda-\lambda_0).
\end{equation}
Since $x_n^\prime(\lambda)>0$, the function $\Re
 S_n(z,\lambda_0)$ has a minimum at $\lambda=\lambda_0$, and
since\\
$\Re S_n(z_n(\lambda_0),\lambda_0)=0$,
$\Re S_n(z_n(\lambda),\lambda_0)\ge 0$ and the equality holds only
at $\lambda=\lambda_0$.

Note that the lower part of $C_n$ differs from the upper one only
by the sign of $y_n(\lambda)$, hence $\Re
S_n(z,\lambda_0)\ge 0$, $z\in C_n$ and the equality holds only at $z=z(\lambda_0)$
and $z=\overline{z(\lambda_0)}$.
\end{proof}
We will prove a similar fact about the behavior of $\Re S_n(z,\lambda_0)$
along the line $L_n:\, \zeta_n(y)=x_n(\lambda_0)+i\,y$.

\begin{lemma}
\label{l:5}  Consider the part of $L_n$, lying in the upper
half-plane $y>0$ . On this part $\Re S_n(z,\lambda_0)=\Re
S_n(\zeta_n(y), \lambda_0)\le 0$ and the equality holds only at
$y=y_n(\lambda_0)$.
\end{lemma}

\begin{proof}
The function $\Re S_n(z,\lambda_0)$ is on $L_n$
\[
\Re S_n(\zeta_n(y),\lambda_0)=\dfrac{x^2_n(\lambda_0)-y^2}{2}+\dfrac{1}{n}
\Re\sumd\limits_{j=1}^n \ln(x_n(\lambda_0)+iy-h_j^{(n)})-\lambda_0 x_n(\lambda_0)+C.
\]
Differentiating this with respect to $y$, we obtain
\begin{equation}
\label{St1}
\Re S_n(\zeta_n(y),\lambda_0)^\prime=y\left(-1+\dfrac{1}{n}
\sumd\limits_{j=1}^n \dfrac{1}{(x_n(\lambda_0)-h_j^{(n)})^2+y^2}\right).
\end{equation}
 Taking into account that the function $\sumd\limits_{j=1}^n
\dfrac{1}{(x_n(\lambda_0)-h_j)^2+y^2}$ is monotone in $y$,
we have from (\ref{cond}) that $y=y_n(\lambda_0)$ is a maximum point
of $\Re S_n(\zeta_n(y),\lambda_0)$. Similarly for $y<0$ the
maximum point is $y=-y_n(\lambda_0)$. Therefore, $\Re
S_n(z,\lambda_0)\le 0$ on $L_n$  and the equality holds only at
$z=z(\lambda_0)$ or $z=\overline{z(\lambda_0)}$.
\end{proof}
Thus, we have proved that
\begin{equation}  \label{Main1}
\Re(n(S_n(t,\lambda_0)-S_n(v,\lambda_0)))\le 0,
\end{equation}
and the equality holds only if $v$ and $t$ are both equal to $z(\lambda_0)$
or $\overline{z(\lambda_0)}$.

We need below also the second derivative of $\Re S_n(z,\lambda_0)$. Assume
that $\lambda\in U_\delta(\lambda_0)$, where $U_\delta(\lambda_0)$ is an
interval $(\lambda_0-\delta,\lambda_0+\delta)$. We get from (\ref{S1})
\begin{equation}  \label{S2}
\Re(- S_n(z_n(\lambda),\lambda_0))^{\prime\prime}=
-x_n^\prime(\lambda)+x_n^{\prime\prime}(\lambda)(\lambda_0-\lambda).
\end{equation}

\begin{lemma}
\label{l:6} There exist $n$-independent $c>0$ and $\delta>0$ such that
\newline
$\Re(- S_n(z_n(\lambda),\lambda_0))^{\prime\prime}<-c$ for any $\lambda\in
U_\delta(\lambda_0)$.
\end{lemma}

\begin{proof}
To prove the lemma it is sufficient to show that
$x_n^{\prime\prime}(\lambda)$ is bounded uniformly in $n$ and
that $x_n^{\prime}(\lambda)$ is bounded from below by a positive
constant uniformly in $n$ in some small enough neighborhood
$U_\delta(\lambda_0)$ of $\lambda_0$. Thus, we will show that $x_n^{\prime}(\lambda)\ge C$ for all $\lambda\in
U_\delta(\lambda_0)$.

   We have from (\ref{zn_pr})
\[
\Re z_n^\prime(\lambda)=x_n^\prime(\lambda)=\Re\left
(\left(1-\dfrac{d}{d\,z}g^{(0)}_n(z_n(\lambda))\right)^{-1}\right)=
\dfrac{a_n(\lambda)}{a^2_n(\lambda)+b^2_n(\lambda)},
\]
where $a_n,\,b_n$ are defined in (\ref{ab}).
  Note that
\[
|b_n(\lambda)|=\left|\dfrac{1}{n}\sumd\limits_{j=1}^n
\dfrac{2y_n(\lambda)(x_n(\lambda)-h_j^{(n)})}
{((x_n(\lambda)-h_j^{(n)})^2+y^2_n(\lambda))^2}\right|\le\dfrac{1}{n}\sumd\limits_{j=1}^n
\dfrac{2|y_n(\lambda)||(x_n(\lambda)-h_j^{(n)})|}
{((x_n(\lambda)-h_j^{(n)})^2+y^2_n(\lambda))^2}\le
\]
\[
\le\dfrac{1}{n}\sumd\limits_{j=1}^n
\dfrac{1}{(x_n(\lambda)-h_j^{(n)})^2+y^2_n(\lambda)}=1.
\]
Hence
\begin{equation}
\label{ner_x}
x_n^{\prime}(\lambda)\ge\dfrac{a_n(\lambda)}{a^2_n(\lambda)+1}.
\end{equation}
Use now the following fact, which will be proved after the
proof of Lemma \ref{l:6}:
\begin{lemma}\label{l:7}
There exist
$n$-independent $C_1$ and $C_2$ such that
\begin{equation}
\label{Ogr_x,y}
|x_n(\lambda)|<C_1,\quad |y_n(\lambda)|<C_1, \quad
|y_n(\lambda)|>C_2,\quad |x_n^{\prime\prime}(\lambda)|<C_1,
\end{equation}
for all $\lambda\in U_\delta(\lambda_0)$, where $n$-independent $\delta$
small enough. Moreover,
\begin{equation}
\label{Ogr_a}
0<c_1<a_n(\lambda)<c_2,\,\,\lambda\in U_\delta(\lambda_0),
\end{equation}
for some $n$-independent $c_1$ and $c_2$.
\end{lemma}

This lemma and (\ref{ner_x}) yield that $x_n^{\prime}(\lambda)\ge
C$ for all $\lambda\in U_\delta(\lambda_0)$ and since
$x_n^{\prime\prime}$ is bounded uniformly, the second terms in
(\ref{S2}) is of order $\delta$. Lemma \ref{l:6} is proved.
\end{proof}

\textbf{Proof of Lemma \ref{l:7}.} We use Lemma \ref{l:3}.
Consider the solution $z(\lambda)$ of the limiting equation
(\ref{eqv_g_0}). Since $ \lambda_0\in \hbox{supp}\,N$, $\Im
z(\lambda_0)=A>0$. Taking into account the continuity of
$z(\lambda)$, we can take a sufficiently small neighborhood
$U_{\delta_1}(\lambda_0)$ such that for $\lambda\in
U_{\delta_1}(\lambda_0)$
\begin{equation}  \label{**}
|z(\lambda)-z(\lambda_0)|<\varepsilon/2.
\end{equation}
Note that we can choose $\lambda_0$-independent $\delta_1$, since $z(\lambda)
$ is uniformly continuous.

Consider the set of the functions $f_\lambda(z)=-g^{(0)}(z)+z-\lambda$ and
the function $\phi(z)=-g^{(0)}_n(z)+g^{(0)}(z)$, where $g^{(0)},\,g^{(0)}_n$
are defined in (\ref{eqv_g_0}),(\ref{g_0,n}), and set $\omega=\{z:
|z-z(\lambda_0)|\le\varepsilon\}$. Let us show that for any $\lambda\in
U_{\delta_1}(\lambda_0)$ and $z\in\partial\omega$
\begin{equation}  \label{Otgr_f}
|f_\lambda(z)|>c_0,
\end{equation}
where $c_0$ does not depend on $\lambda$. Assume the opposite and
choose a sequence $\{\lambda_k\}_{k\ge 1}$,$\lambda_k\in
U_{\delta_1}(\lambda_0)$ such that $|f_{\lambda_k}(z_k)|\to 0$, as
$k\to \infty$. There exists a subsequence $\{\lambda_{k_m}\}$,
converging to some $\lambda\in U_{\delta_1}(\lambda_0)$ such that
the subsequence $\{z_{k_m}\}$ converges to $z\in\partial\omega$.
For these $\lambda$ and $z$ $f_{\lambda}(z)=0$. But equation
$f_{\lambda}(z)=0$ has in the upper half-plane only one root $
z(\lambda)$, which is inside of the circle of the radius
$\varepsilon/2$ and with the center $z(\lambda_0)$. This
contradiction proves (\ref{Otgr_f}). Since $g^{(0)}_n(z)\to
g^{(0)}(z)$ uniformly on any compact set of the upper half-plane
(recall weak convergence $N_n^{(h)}\to N^{(h)}$), we have starting
from some $n$
\begin{equation}  \label{Ogr_phi}
|\phi(z)|<c_0,\,\, z\in \partial\omega
\end{equation}
Comparing (\ref{Otgr_f}) and (\ref{Ogr_phi}), we obtain that starting from
some $n$
\begin{equation*}
|f_\lambda(z)|>|\phi(z)|, \,\,z\in\partial\omega,\,\, \forall\lambda\in
U_{\delta_1}(\lambda_0).
\end{equation*}
Since both functions are analytic, the Rouchet theorem implies
that $ f_\lambda(z)$ and $f_\lambda(z)+\phi(z)=z-
g^{(0)}_n(z)-\lambda$ have the same number of zeros in $\omega$.
Since $f_\lambda(z)$ has only one zero in $ \omega$, we conclude
that $z_n(\lambda)$ belongs to $\omega$, $x_n(\lambda)$ and
$y_n(\lambda)$ are bounded and $y_n(\lambda)>C>0$ uniformly in $n$
if $ \lambda\in U_\delta(\lambda_0)$, where $\delta$ one can take
equal to $ \delta_1$. Since $z_n(\lambda)$ is analytic, we proved
also that $ x_n^{\prime\prime}(\lambda)$ is bounded uniformly in
$n$ if $\lambda\in U_\delta(\lambda_0)$.

Note that we have proved also that for any $\lambda_0$ such that $
\rho(\lambda_0)>0$ and for any $\varepsilon>0$ there exists
$\delta$ such that for any $\lambda\in U_\delta(\lambda_0)$ and
any $n>N(\delta, \varepsilon)$
\begin{equation*}
|z_n(\lambda)-z(\lambda)|\le 2\varepsilon.
\end{equation*}
Observe also that we can take an interval
$(a,b)\subset\hbox{supp}\,N$ such that $\lambda_0\in (a,b)$ and
for all $\lambda\in (a,b)$ $\rho(\lambda)=\pi \Im g(\lambda+i\cdot
0)=\Im z(\lambda)>0$. Thus, we proved that $ z_n(\lambda)\to
z(\lambda)$, $n\to \infty$ uniformly in $\lambda\in (a,b)$.

Since $g^{(0)}_n$ is analytic,
$\displaystyle\frac{d}{d\,z}g^{(0)}_n\to
\displaystyle\frac{d}{d\,z} g^{(0)}$ also uniformly on any compact
set of the upper half-plane. Recall that
$a_n(\lambda)=\Re\left(1-\displaystyle
\frac{d}{d\,z}g^{(0)}_n(z_n(\lambda))\right)$. Since
$z_n(\lambda)\in \omega$ if $\lambda\in U_\delta(\lambda_0)$, it
suffices to prove (\ref{Ogr_a}) for
\begin{equation*}
\Re\left(1-\displaystyle\frac{d}{d\,z}g^{(0)}(z_n(\lambda))\right)=
\displaystyle\int\displaystyle\frac{2y_n^2(\lambda)N^{(0)}(d\,h)}
{ ((x_n(\lambda)-h)^2+y_n^2(\lambda))^2}.
\end{equation*}
But if for $\lambda\in U_\delta(\lambda_0)$ $x_n(\lambda)$ and $y_n(\lambda)$
are bounded, $y_n(\lambda)>C>0$ uniformly in $n$, and $\hbox{supp} \,N^{(0)}$
is bounded, the r.h.s. here is bounded from both sides by some positive
constants. $\quad\Box$

\medskip

According to Lemma \ref{l:6} and by the hypothesis of the Theorem \ref{thm:1}
\begin{equation}  \label{Okr}
\Re\left(-
S_n(z_n(\lambda),\lambda_0)\right)<-c\displaystyle\frac{
(\lambda-\lambda_0)^2}{2}, \,\,\lambda\in U_\delta(\lambda_0).
\end{equation}
Since
$\displaystyle\frac{d}{d\,\lambda}\Re(S_n(z_n(\lambda),\lambda_0))$
has the unique root $\lambda=\lambda_0$, the function
$\Re\left(S_n(z_n( \lambda),\lambda_0)\right)$ is monotone for
$\lambda\ne\lambda_0$ and we have outside of $U_\delta(\lambda_0)$
\begin{equation}  \label{vne_okr}
\Re(- S_n(z_n(\lambda),\lambda_0))<-c\displaystyle\frac{\delta^2}{2}.
\end{equation}
Apply analogous argument to the neighborhood of $z_n(\lambda_0)$ on $L_n$.
We have from (\ref{St1})
\begin{equation}  \label{St2}
\begin{array}{c}
\Re(S_n(z_n(y),\lambda_0))^{\prime\prime}=-1+\displaystyle\frac{1}{n}
\displaystyle\sum\limits_{j=1}^n \displaystyle\frac{1}{(x_n(
\lambda_0)-h_j^{(n)})^2+y^2}
-\displaystyle\frac{1}{n}\displaystyle \sum\limits_{j=1}^n
\displaystyle\frac{2y^2}{((x_n(
\lambda_0)-h_j^{(n)})^2+y^2)^2} \\
=\displaystyle\frac{1}{n} \displaystyle\sum\limits_{j=1}^n
\displaystyle
\frac{y_n^2(\lambda_0)-y^2}{((x_n(\lambda_0)-h_j^{(n)})^2+y^2)
\cdot((x_n(\lambda_0)-h_j^{(n)})^2+y^2_n(\lambda_0))} \\
-\displaystyle\frac{1}{n}\displaystyle\sum\limits_{j=1}^n
\displaystyle\frac{ 2y^2}{((x_n(\lambda_0)-h_j^{(n)})^2+y^2)^2}
\end{array}
\end{equation}
Consider $y\in U_{\delta/2}(y(\lambda_0))$ ($y(\lambda_0)>0$) and recall
that $y_n(\lambda_0)\in U_{\delta/2}(y(\lambda_0))$ starting from some $n$.
Hence, if $n$ big enough
\begin{equation*}
|y_n(\lambda)-y|<\delta.
\end{equation*}
This and (\ref{St2}) yield
\begin{equation*}
\Re(S_n(\zeta_n(y),\lambda_0))^{\prime\prime}<-c,\,\,\hbox{if}\,\, y\in
U_{\delta/2}(y(\lambda_0)),
\end{equation*}
hence,
\begin{equation}  \label{Okr_y}
\Re(S_n(\zeta_n(y),\lambda_0))<-c\displaystyle\frac{(y-y_n(\lambda_0))^2}{2}.
\end{equation}
Since $\displaystyle\frac{d}{d\,y}\Re(S_n(\zeta_n(y),\lambda_0))$
has the unique root $y=y_n(\lambda_0)$, the function
$\Re(S_n(\zeta_n(y),\lambda_0))$ is monotone for $y\ne
y_n(\lambda_0)$ and we have outside of $
U_{\delta/2}(y(\lambda_0))$
\begin{equation}  \label{vne_okr_y}
\Re(S_n(\zeta_n(y),\lambda_0))<-c\displaystyle\frac{\delta^2}{2}.
\end{equation}
Besides, since $\displaystyle\frac{d^2}{d\,y^2}\Re(S_n(\zeta_n(y),
\lambda_0))\to -1$, as $y\to \infty$, uniformly in $n$,
$\Re(S_n(\zeta_n(y), \lambda_0))$ is convex. Hence we get for some
fixed segment $[-K;K]$ (we can take $n$-independent $K$, taking
into account that $z_n(\lambda_0)$ is in some neighborhood of
$z(\lambda_0)$)
\begin{equation}  \label{y_infty}
\Re(S_n(\zeta_n(y),\lambda_0))<-c_1|y|+c_2, \,\, c_1>0.
\end{equation}
Denote $U_1=U_\delta(\lambda_0)$,  $U_2=U_\delta(y(\lambda_0))$.
Using formulas (\ref{Okr}),(\ref{vne_okr}),\quad (\ref{Okr_y}) and
(\ref{vne_okr_y} ), we obtain for sufficiently big $n$
\begin{equation}  \label{Int1}
\begin{array}{c}
\left|\displaystyle\int\limits_{L_n} \displaystyle\frac{dt}{2\pi}
\oint\limits_{C_n}\displaystyle\frac{dv}{2\pi}\exp\{v\mu^\prime-
t\lambda^\prime\}\displaystyle\frac{\exp\{n(S_n(t,\lambda_0)-
S_n(v,\lambda_0))\}} {v-t}\right| \\
\\
\le
C\left(\displaystyle\int\limits_{U_2}\displaystyle\int\limits_{U_1}+
\displaystyle\int\limits_{U_2}\oint\limits_{C_n\backslash U_1}+
\displaystyle \int\limits_{L_n\backslash U_2}
\displaystyle\int\limits_{U_1}\right)
\displaystyle\frac{\exp\{\Re(n(S_n(\zeta_n(y),\lambda_0)-
S_n(z_n(\lambda),\lambda_0)))\}|z_n^\prime|d\,\lambda\,
d\,y}{|z_n(\lambda)-
\zeta_n(y)|} \\
\le C \displaystyle\int\limits_{U_2}\displaystyle\int\limits_{U_1}
\displaystyle\frac{|z_n^\prime(\lambda)|d\,\lambda
dy}{|z_n(\lambda)- \zeta_n(y)|}+C_1\cdot
|C_n|\cdot\exp\{-c\displaystyle\frac{n\delta^2}{2} \}+C_2\cdot
\exp\{-c\displaystyle\frac{(n-1)\delta^2}{2}\},
\end{array}
\end{equation}
where $|C_n|$ is the length of the contour $C_n$. Note that
\begin{equation*}
\displaystyle\int\limits_{U_2}\displaystyle\int\limits_{U_1}
\displaystyle \frac{|z_n^\prime(\lambda)|d\,\lambda
dy}{|z_n(\lambda)-\zeta_n(y)|}\le
\displaystyle\int\limits_{U_2}\displaystyle\int\limits_{U_1}
\displaystyle \frac{|z_n^\prime(\lambda)|d\,\lambda
dy}{\sqrt{(1-\cos \alpha_n+o(\delta))(|z_n(\lambda)|^2+
|\zeta_n(y)|^2)}},
\end{equation*}
where $\alpha_n$ is the angle between $C_n$ and $L_n$ at the point
$ z(\lambda_0)$, i.e.,
$\cot\alpha_n=\displaystyle\frac{y_n^\prime(\lambda_0)}{
x_n^\prime(\lambda_0)}$. Since $x_n^\prime(\lambda_0)>c>0$, $\cos
\alpha_n<1-\varepsilon$, we have
\begin{equation}  \label{I_okr}
\displaystyle\int\limits_{U_2}\displaystyle\int\limits_{U_1}
\displaystyle \frac{|z_n^\prime(\lambda)|d\,\lambda
dy}{\sqrt{(1-\cos \alpha_n+o(\delta))(|z_n(\lambda)|^2+
|\zeta_n(y)|^2)}}\le C_0\displaystyle
\int\limits_{U_2}\displaystyle\int\limits_{U_1}
\displaystyle\frac{ |z_n^\prime(\lambda)|d\,\lambda
dy}{\sqrt{|z_n(\lambda)|^2+ |\zeta_n(y)|^2}} \le C\cdot 4\delta.
\end{equation}

Now we need the following

\begin{lemma}
\label{l:8} The length $|C_n|$ of the contour $C_n$ admits the bound:
\begin{equation*}
|C_n|\le C n.
\end{equation*}
\end{lemma}

\begin{proof} We will find the bound for the length of the
part of $C_n$ between the lines $x=x_1$ and
$x=x_2$, $x_2-x_1=2$. Denote
\begin{equation}\label{obozn}
\begin{array}{c}
y^2(x)=s(x),\quad x-h_j^{(n)}=\triangle_j,\,\,\\
\sigma_k= \dfrac{1}{n}\sumd\limits_{j=1}^n
 \dfrac{1}{(\triangle_j^2+s)^k},\quad
\sigma_{kl}= \dfrac{1}{n}\sumd\limits_{j=1}^n
\dfrac{\triangle^l_j}{(\triangle_j^2+s)^k}\,\, k=\overline{1,3},\,l=1,2.\\
\end{array}
\end{equation}
Differentiating (\ref{cond}) with respect to $x$, we obtain the
equality
\[
-s^\prime\dfrac{1}{n}\sumd\limits_{j=1}^n\dfrac{1}
{(\triangle_j^2+s)^2}-\dfrac{2}{n}\sumd\limits_{j=1}^n\dfrac{\triangle_j}
{(\triangle_j^2+s)^2}=0,
\]
implying that
\begin{equation}\label{s1}
|s^\prime|=2|\sigma_{21}|\sigma_2^{-1}\le
2\sigma_{22}^{1/2} \sigma_2^{-1/2}\le 2\sigma_2^{-1/2}\le 2\sigma_1^{-1/2}=2.
\end{equation}
Differentiating (\ref{cond}) with respect to $x$ twice, we have
\begin{multline}
s^{\prime\prime}\cdot\left(\dfrac{1}{n}\sumd\limits_{j=1}^n
 \dfrac{1}{(\triangle_j^2+s)^2}\right)-2(s^\prime)^2\cdot
\left(\dfrac{1}{n}\sumd\limits_{j=1}^n
 \dfrac{1}{(\triangle_j^2+s)^3}\right)\\
-8s^\prime\cdot
\left(\dfrac{1}{n}\sumd\limits_{j=1}^n
\dfrac{\triangle_j}{(\triangle_j^2+s)^3}\right)+
\dfrac{2}{n}\sumd\limits_{j=1}^n
 \dfrac{(\triangle_j^2+s)^2-4\triangle_j^2(\triangle_j^2+s)}
 {(\triangle_j^2+s)^4}=0,
\end{multline}
or, in our notations
\begin{equation}\label{ur_s2}
s^{\prime\prime}\sigma_2-2(s^\prime)^2\sigma_3-8s^\prime\sigma_{31}+
2(4s\sigma_3-3\sigma_2)=0.
\end{equation}
Note that
\[
s\sigma_3=\dfrac{1}{n}\sumd\limits_{j=1}^n
 \dfrac{s}{(\triangle_j^2+s)^3}\le \dfrac{1}{n}\sumd\limits_{j=1}^n
 \dfrac{1}{(\triangle_j^2+s)^2}=\sigma_2,
\]
and also
\[
\sigma_{31}^2=\left(\dfrac{1}{n}\sumd\limits_{j=1}^n
\dfrac{\triangle_j}{(\triangle_j^2+s)^3}\right)^2\le
\dfrac{1}{n}\sumd\limits_{j=1}^n
\dfrac{\triangle_j^2}{(\triangle_j^2+s)^3}\cdot
\dfrac{1}{n}\sumd\limits_{j=1}^n
\dfrac{1}{(\triangle_j^2+s)^3}\le\sigma_2\sigma_3.
\]
Using this inequality, we get from (\ref{ur_s2})
\[
\begin{array}{c}
s^{\prime\prime}\sigma_2=2(s^\prime)^2\sigma_3+8s^\prime\sigma_{31}-
2(4s\sigma_3-3\sigma_2)=2\sigma_3\left(s^\prime+2\sigma_{31}/\sigma_3\right)^2
-8\sigma^2_{31}/\sigma_3\\
-8s\sigma_3+6\sigma_2\ge -8\sigma^2_{31}/\sigma_3-2\sigma_2\ge -10\sigma_2,\\
\end{array}
\]
or
\begin{equation}\label{s2}
s^{\prime\prime}\ge -10.
\end{equation}
Let $x_*\in [x_1;x_2]$ be the maximum point of $y(x)$, and
$y^\prime(x)=\frac{s^\prime(x)}{2\sqrt{s(x)}}>0$ when $x\in[x_0,x_*]$
and let $l(x)$ be the length of $C_n$ between $x_1$ and $x\in
[x_1;x_2]$. Then we have
\begin{multline}\label{ots_l1}
l(x_*)-l(x_0)=\intd\limits_{x_0}^{x_*}\sqrt{1+(y^\prime(x))^2}d\,x=
\intd\limits_{x_0}^{x_*}\sqrt{1+\left(\dfrac{s^\prime(x)}{2\sqrt{s(x)}}\right)^2}d\,x
\\
\le\intd\limits_{x_0}^{x_*}\left(1+
\dfrac{s^\prime(x)}{2\sqrt{s(x)}}\right)d\,x=(x_*-x_0)+\sqrt{s_*}-\sqrt{s_0}
\le(x_*-x_0)+\sqrt{s_*-s_0},
\end{multline}
where $s_*=s(x_*)$, $s_0=s(x_0)$.
Taking into account that $s^\prime(x_*)=0$, we write
\[
s_0-s_*=\dfrac{s^{\prime\prime}(\xi)(x_0-x_*)^2}{2},
\]
where $\xi\in [x_0,x_*]$.
This and (\ref{s2}) imply
\[
0\le s_*-s_0\le 5(x_0-x_*)^2.
\]
Hence, we get in view of (\ref{ots_l1})
\begin{equation}\label{ots_l2}
l(x_*)-l(x_0)\le (1+\sqrt{5})(x_*-x_0).
\end{equation}
We have similar inequality for $x_0>x_*$ and $y^\prime(x)<0$,
$x\in [x_*,x_0]$. Take an arbitrary $x_0\in [x_1;x_2]$ and denote
$x_*$ the nearest to $x_0$ maximum point of $y(x)$ in $[x_1,x_0]$.
Then, splitting $[x_1,x_*] $ in the segments of monotonicity of $y$
and using
(\ref{s1}), (\ref{ots_l2}), and its analog for decreasing $y(x)$, we obtain
\begin{multline}\label{ots_l}
l(x_0)=l(x_*)+\intd\limits_{x_*}^{x_0}l^\prime(x)d\,x
\le (1+\sqrt{5})(x_*-x_1)+\intd\limits_{x_*}^{x_0}\left(1+
\dfrac{|s^\prime(x)|}{2\sqrt{s(x)}}\right)d\,x\\
\le (1+\sqrt{5})(x_*-x_1)+(x_0-x_*)+\sqrt{|s_0-s_*|}\\
\le(1+\sqrt{5})(x_*-x_1)+(x_0-x_*)+\sqrt{2}\sqrt{x_0-x_*}\le
C\sqrt{x_0-x_1},
\end{multline}
where the last inequality holds, because $|x_0-x_*|\le |x_0-x_1|$
and $|x_0-~x_1|\le~2$.
Hence,
\[
l(x_2)\le C \sqrt{x_2-x_1}\le C.
\]
It follows from (\ref{cond}) that
$\hbox{dist}(x_n(\lambda),\{h_j^{(n)}\}_{j=1}^n)\le 1$. Therefore,
we can cover $C_n$ by the $n$ stripes of the
width 2 and thus we obtain that $|C_n|\le Cn$.
\end{proof}
Using Lemma \ref{l:8}, (\ref{I_okr}) and (\ref{Int1}) we  get that
\begin{equation}  \label{Int}
\lim\limits_{\delta\to 0}\displaystyle\int\limits_{L_n}
\displaystyle\frac{dt
}{2\pi}\oint\limits_{C_n}\displaystyle\frac{dv}{2\pi}\exp\{v\mu^\prime-
t\lambda^\prime\}\displaystyle\frac{\exp\{n(S_n(t,\lambda_0)-
S_n(v,\lambda_0))\}} {v-t}=0.
\end{equation}
Recall that
\begin{equation*}
K_n(\lambda,\mu)=-n\displaystyle\int\limits_{L}\displaystyle\frac{dt}{2\pi}
\oint\limits_{C_n}
\displaystyle\frac{dv}{2\pi}\exp\{v\lambda^\prime-t\mu^ \prime\}
\displaystyle\frac{\exp\{n(S_n(t,\lambda_0)- S_n(v,\lambda_0))\}}{
v-t}.
\end{equation*}
Change the order of the integrations and move the integration over
$t$ from $ L$ to $L_n$. To this end consider the contour
$C_{R,\varepsilon}$ of Fig.1, where $R$ is big enough

\includegraphics[width=4 in, height=2.5 in]{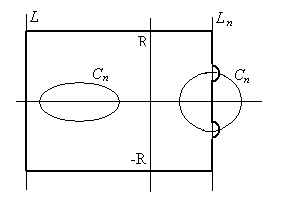} 
\label{Fig.1}

It is clear that the integral with respect to $t$ over this contour is equal
to the residue at $v=t$ for any $v$ between $L$ and $L_n$:
\begin{equation*}
\displaystyle\int\limits_{C_{R,\varepsilon}}\displaystyle\frac{dt}{2\pi}
\exp\{v\lambda^\prime-t\mu^\prime\}
\displaystyle\frac{\exp\{n(S_n(t, \lambda_0)-
S_n(v,\lambda_0))\}}{v-t}= i \cdot
\exp\{v(\mu^\prime-\lambda^\prime)\}.
\end{equation*}
If $v$ does not lie between $L$ and $L_n$, then we can find
$\delta$ such that $v$ is inside of the contour
$C_{R,\varepsilon}$ for any $ \varepsilon<\delta$. Therefore, we
have for sufficiently big $R$ and for $ \varepsilon\to 0$
\begin{equation*}
\begin{array}{c}
\lim\limits_{\varepsilon\to
0}\displaystyle{\oint\limits_{C_n}}\displaystyle
\frac{dv}{2\pi}\displaystyle\int\limits_{C_{R,\varepsilon}}\displaystyle
\frac{dt}{2\pi} \exp\{v\lambda^\prime-t\mu^\prime\}
\displaystyle\frac{
\exp\{n(S_n(t,\lambda_0)- S_n(v,\lambda_0))\}}{v-t} \\
=-\displaystyle\frac{i}{2\pi}\,\displaystyle\int\limits_{\overline{
z_n(\lambda_0)}}^{z_n(\lambda_0)}
\exp\{v(\lambda^\prime-\mu^\prime)\} dv=
\exp\{x_n(\lambda_0)(\lambda^\prime-
\mu^\prime)\}\displaystyle\frac{
\sin(y_n(\lambda_0)(\lambda^\prime-\mu^\prime))} {\pi
(\lambda^\prime-\mu^\prime)}.
\end{array}
\end{equation*}
Integrals over the lines $\Im z=\pm R$ have the order $C\,e^{-nR^2/2}$, and
we get for $R\to\infty$
\begin{equation}  \label{Res}
\begin{array}{c}
\displaystyle{\oint\limits_{C_n}}\displaystyle\frac{dv}{2\pi}
\oint\limits_{L\,\bigcup\,L_n} \displaystyle\frac{dt}{2\pi}
\exp\{v\lambda^\prime-t\mu^\prime\}
\displaystyle\frac{\exp\{n(S_n(t,
\lambda_0)- S_n(v,\lambda_0))\}}{v-t} \\
= \exp\{x_n(\lambda_0)(\lambda^\prime-
\mu^\prime)\}\displaystyle\frac{
\sin(y_n(\lambda_0)(\lambda^\prime-\mu^\prime))} {\pi
(\lambda^\prime-\mu^\prime)}.
\end{array}
\end{equation}
Thus, adding (\ref{Res}) and (\ref{Int1}), we obtain
\begin{equation}  \label{Main}
\begin{array}{c}
\displaystyle\frac{1}{n}K_n(\lambda,\mu)=-\displaystyle\int\limits_{L}
\displaystyle\frac{dt}{2\pi}\oint\limits_{C_n}
\displaystyle\frac{dv}{2\pi} \exp\{v\lambda^\prime-t\mu^\prime\}
\displaystyle\frac{\exp\{n(S_n(t,
\lambda_0)- S_n(v,\lambda_0))\}}{v-t} \\
=-\left(\displaystyle\int\limits_{L_n}\displaystyle\frac{dt}{2\pi}
\oint\limits_{C_n}
\displaystyle\frac{dv}{2\pi}+\oint\limits_{L\,\cup L_n}
\displaystyle\frac{dt}{2\pi}\oint\limits_{C_n}
\displaystyle\frac{dv}{2\pi}
\right)\exp\{v\lambda^\prime-t\mu^\prime\} \displaystyle\frac{
\exp\{n(S_n(t,\lambda_0)- S_n(v,\lambda_0))\}}{v-t} \\
=-\displaystyle\int\limits_{L_n} \displaystyle\frac{dt}{2\pi}
\oint\limits_{C_n}\displaystyle\frac{dv}{2\pi}\exp\{v\lambda^\prime-
t\mu^\prime\}\displaystyle\frac{\exp\{n(S_n(t,\lambda_0)-
S_n(v,\lambda_0))\}
} {v-t} \\
\end{array}
\end{equation}
\begin{equation*}
\begin{array}{c}
+\exp\{x_n(\lambda_0)(\lambda^\prime-
\mu^\prime)\}\displaystyle\frac{
\sin(y_n(\lambda_0)(\lambda^\prime-\mu^\prime))} {\pi
(\lambda^\prime-\mu^\prime)} \\
=\exp\{x_n(\lambda_0)(\lambda^\prime-
\mu^\prime)\}\displaystyle\frac{\sin
(y_n(\lambda_0)(\lambda^\prime-\mu^\prime))}
{\pi(\lambda^\prime-\mu^\prime)}
+o(1),\,\,n\to\infty. \\
\end{array}
\end{equation*}
Note that in the proof of Lemma \ref{l:7} we have shown that $
z_n(\lambda)\to z(\lambda)$ as $n\to\infty$ uniformly in $\lambda
\in (a,b)\subset \hbox{supp}\,N$, where $z_n(\lambda)$ and
$z(\lambda)$ are the solutions of equation (\ref{z}) and
(\ref{eqv_g_0}). Hence we have $
\lim\limits_{n\to\infty}y_n(\lambda)=y(\lambda)=
\pi\rho(\lambda)>0$ uniformly in $\lambda\in (a,b)$. Besides, it
follows from (\ref{Main}) that for
$\rho_n(\lambda)=\displaystyle\frac{1}{n}K_n(\lambda,\lambda)$ the
inequality
$|\displaystyle\frac{1}{n}K_n(\lambda,\lambda)-y_n(\lambda)|<
\varepsilon$ holds uniformly in $\lambda\in (a,b)$, since all
bounds were $ \lambda$-independent. Therefore we have proved that
$\rho_n(\lambda)\to \rho(\lambda)$, as $n\to\infty$, uniformly in
$\lambda\in (a,b)$. Now we obtain (\ref{Un}) by using (\ref{Det})
and (\ref{Main}).

\section{Proof of the Theorem \protect\ref{thm:2}.}

We start from the following

\begin{lemma}
\label{l:9} Let $g^{(0)}_n$ and $g^{(0)}$ be defined in
(\ref{g_0,n}),(\ref {eqv_g_0}). Then we have under conditions of
Theorem \ref{thm:2}
\begin{equation}  \label{rav_sh}
\lim\limits_{n\to
\infty}\mathbf{P}\{|g^{(0)}_n(z)-g^{(0)}(z)|>\varepsilon \}= 0
\end{equation}
uniformly in $z$ from compact set $K$ in the upper half-plane.
\end{lemma}

\begin{proof}
Note that it suffices to prove (\ref{rav_sh}) for any $z\in K$.
Indeed, let $\{z_j\}_{j=1}^l$ be a $\varepsilon$-net of the
compact set $K$. Then there exists $N$ such that for any $n>N$ and
for any $\delta>0$
\[
{\bf P}\{\bigcup\limits_{j=1}^l\{|g^{(0)}_n(z_j)-g^{(0)}(z_j)|>\varepsilon\}\}
\le \sum\limits_{j=1}^l {\bf P}\{|g^{(0)}_n(z_j)-g^{(0)}(z_j)|>\varepsilon\}
<\delta.
\]
Besides, for any $z\in K$ there exists $z_k\in \{z_j\}_{j=1}^l$ such that
$|z-z_k|<\varepsilon$. Therefore since $\left|\dfrac{d}{d\,z}
g^{(0)}_n\right|\le 1/\Im^2z$, $\left|\dfrac{d}{d\,z}g^{(0)}\right|\le 1/\Im^2 z$
\[
|g^{(0)}_n(z)-g^{(0)}(z)|\le |g^{(0)}_n(z_k)-g^{(0)}(z_k)|+2\varepsilon/\Im^2z.
\]
Hence, taking into account that $\Im z$ is bounded from below
by a positive constant for $z\in K$, we have for any $n>N$
\[
{\bf P}\{|g^{(0)}_n(z)-g^{(0)}(z)|<C\varepsilon\}>1-\delta.
\]
We are left to prove that (\ref{rav_sh}) is valid pointwise.
Since
\[
\intd \lambda^2d\,N^{(0)}_n(\lambda)<\infty,
\]
there exists $A$ such that
\begin{equation}\label{hvost}
\intd\limits_{|\lambda|>A}d\,N^{(0)}_n(\lambda)\le \dfrac{1}{A^2}\intd
\lambda^2d\,N^{(0)}_n(\lambda)<\varepsilon.
\end{equation}
Set
\[
f(\lambda)=\dfrac{1}{\lambda-z} \quad(\lambda\in \mathbb{R}),
\quad f_A(\lambda)=\left\{\begin{array}{ll}
\dfrac{1}{\lambda-z},\quad&\lambda\in [-A,A],\\
0,\quad&\lambda\not\in [-A,A],
\end{array}\right.
\]
and let $f^\varepsilon$ be a piecewise constant function on the
segment [-A,A] such that
\begin{equation}\label{stup}
|f^\varepsilon(\lambda)-f_A(\lambda)|<\varepsilon.
\end{equation}
If $f^\varepsilon(\lambda)=f_j,\,\lambda\in \triangle_j,\,
j=\overline{1,s}$, then we have from (\ref{hvost})
\begin{multline}\label{obr}
|g^{(0)}_n(z)-g^{(0)}(z)|\le \left|\intd f(\lambda)d\,N^{(0)}_n(\lambda)
-\intd f_A(\lambda)d\,N^{(0)}_n(\lambda)\right| +\\
\left|\intd f_A(\lambda)d\,N^{(0)}_n(\lambda)-
\intd f_A(\lambda)d\,N_{0}(\lambda)\right|+
\left|\intd f_A(\lambda)d\,N_{0}(\lambda)\right.\\-
\left.\intd f(\lambda)d\,N_{0}(\lambda)\right|\le C\varepsilon +
\left|\intd f_A(\lambda)d\,N^{(0)}_n(\lambda)-
\intd f_A(\lambda)d\,N_{0}(\lambda)\right|
\end{multline}
Besides, it follows from (\ref{stup}) that
\begin{multline}\label{apr}
\left|\intd f_A(\lambda)d\,N^{(0)}_n(\lambda)-\intd
f_A(\lambda)d\,N_{0}(\lambda)\right|\le
\left|\intd f_A(\lambda)d\,N^{(0)}_n(\lambda)
\right.\\-
\left.\intd f^\varepsilon(\lambda)d\,N^{(0)}_n(\lambda)\right| +
\left|\intd f^\varepsilon(\lambda)d\,N^{(0)}_n(\lambda)-
\intd f^\varepsilon(\lambda)d\,N_{0}(\lambda)\right|+
\left|\intd f^\varepsilon(\lambda)d\,N_{0}(\lambda)\right.\\-
\left.\intd f_A(\lambda)d\,N_{0}(\lambda)\right|\le 2\varepsilon +
\left|\intd f^\varepsilon(\lambda)d\,N^{(0)}_n(\lambda)-
\intd f^\varepsilon(\lambda)d\,N_{0}(\lambda)\right|
\end{multline}
We have also that
\begin{equation}\label{kus_pos}
\left|\intd f^\varepsilon(\lambda)d\,N^{(0)}_n(\lambda)-
\intd
f^\varepsilon(\lambda)d\,N_{0}(\lambda)\right|=
\sum\limits_{j=1}^s
f_j\cdot|N^{(0)}_n(\triangle_j)-N^{(0)}(\triangle_j)|,
\end{equation}
and by the condition of Theorem \ref{thm:2}, for any $\delta$
there exists $N$ such that for any $n>N$
\begin{equation}\label{ver_N}
{\bf P}\{\bigcup\limits_{j=1}^l\{|N^{(0)}_n(\triangle_j)-
N^{(0)}(\triangle_j)|>\varepsilon\}\}<\delta.
\end{equation}
Now the assertion of lemma follows from (\ref{obr}),(\ref{apr}),
(\ref{kus_pos}), and (\ref{ver_N}).
\end{proof}

Let us take the disk
$\omega=\{z:\,|z(\lambda_0)-z|\le\varepsilon\}$ as the compact set
$K$. Taking into account (\ref{rav_sh}), we have that for any
small $\delta$ there exists $N$ such that for all $n>N$ the set of
events $ \Omega_{\varepsilon}$ such that
\begin{equation*}
|g^{(0)}_n(z)-g^{(0)}(z)|<\varepsilon\ \,\,(z\in \omega),
\end{equation*}
satisfies the condition $\mathbf{P}\{\Omega_\varepsilon\}\ge 1-\delta$.

We want to find for any $m$
\begin{multline}  \label{limr}
\lim\limits_{n\to
\infty}\displaystyle\frac{1}{(n\rho_n(\lambda_0))^m} R^{(n)}_m
\left(\lambda_0+\displaystyle\frac{x_1}{n\rho_n(\lambda_0)},
\ldots,\lambda_0+\displaystyle\frac{x_m}{n\rho_n(\lambda_0)}\right) \\
= \lim\limits_{n\to \infty}
\mathbf{E}^{(h)}\left\{\det\left\{\displaystyle
\frac{1}{n\rho_n(\lambda_0)}K_n\left(\lambda_0+\displaystyle\frac{x_i}{
n\rho_n(\lambda_0)},\lambda_0+
\displaystyle\frac{x_j}{n\rho_n(\lambda_0)}
\right)\right\}\right\}.
\end{multline}
Note that the argument used in the proof of Theorem \ref{thm:1}
remains valid for all events from $\Omega_{\varepsilon}$. Using
the uniform bound for $\frac{1}{n}K_n(\lambda,\lambda)$ which will
be proved below (see Lemma \ref{l:10}) we can see that the
contribution from $\Omega\setminus \Omega_{\varepsilon}$ can be
bounded by $C\delta$. So, we can divide by $
\rho_n(\lambda_0)=\frac{1}{n}\mathbf{E}^{(h)}\{K_n(\lambda_0,\lambda_0)\}$.

Choose small $\varepsilon$ and $\delta$ and split
$\mathbf{E}^{(h)}\{\ldots\} $ in (\ref{limr}) into two parts: the
integral over $\Omega_{\varepsilon}$ and the integral over its
complement. We can repeat the arguments used in the proof of
Theorem \ref{thm:1} for the integral over $\Omega_{\varepsilon}$
to obtain the property (\ref{Un}). To bound the integral over the
complement of $\Omega_{\varepsilon}$ we use

\begin{lemma}
\label{l:10} We have for any set $\{h_j^{(n)}\}_{j=1}^n$ and for
any $ \lambda=\lambda_0+\lambda^\prime/n$
\begin{equation}  \label{ogr_K}
\left|\displaystyle\frac{1}{n}K_n(\lambda_0+\lambda^\prime/n,
\lambda_0+\lambda^\prime/n)\right|\le C,
\end{equation}
where $K_n$ is defined in (\ref{K}).
\end{lemma}

\begin{proof}
As in the proof of Theorem \ref{thm:1}, take $C_n$ as a contour $C$
and move the integration with respect to $t$ from $L$ to $L_n$.
Using (\ref{Res}) as in (\ref{Main}) we obtain
\begin{equation}
\label{1}
\dfrac{1}{n}K_n(\lambda, \lambda)=-\intd\limits_{L_n}
\dfrac{dt}{2\pi}\oint\limits_{C_n}\dfrac{dv}{2\pi}\exp\{\lambda^\prime(v-t)
\}\dfrac{\exp\{n(S_n(t,\lambda_0)- S_n(v,\lambda_0))\}}
{v-t}
-\dfrac{y_n(\lambda_0)}{\pi}.
\end{equation}
If $y_n(\lambda)\neq 0$, then (\ref{cond}) implies that
$|y_n(\lambda)|\leq 1$, thus $y_n(\lambda)$ is bounded uniformly in $n$
for any $\lambda$, in particle, for $\lambda=\lambda_0$. Hence, to
prove the lemma it is necessary and sufficient to check the
uniform bound for the double integral in (\ref{1}).

We need the following
\begin{lemma}\label{l:11}
Let $J=[x_n(\lambda_0);x_n(\lambda_1)]$ moreover $|J|=1$.
Then there exists $n$-independent constant $\delta$, such that
\begin{equation}
\label{2}
|\Re S_n(z_n(\lambda_1),\lambda_0)- \Re S_n(z_n(\lambda_0),\lambda_0)|
\ge \delta/\ln^{12}n.
\end{equation}
\end{lemma}
The lemma will be proved after the proof of Lemma \ref{l:10}.

  Consider  the integral in (\ref{1}). Let
$I=[x_n(\lambda_1),x_n(\lambda_2)]$ be  a segment such that
$|x_n(\lambda_1)-x_n(\lambda_0)|=|x_n(\lambda_2)-x_n(\lambda_0)|=1$.
Since $\Re S_n(z_n(\lambda_0),\lambda_0)=0$, according to Lemma
\ref{l:11} we have
$$
\Re S_n(z_n(\lambda),\lambda_0)\le -\delta/\ln^{12} n
$$
outside of $I$ for some $n$-independent $\delta>0$.
Hence, since length of $C_n$ is $O(n)\,\,(n\to \infty)$
(see Lemma \ref{l:8}) and  the integral with respect to $t$ is
bounded, the whole integral over this part of $C_n$ is
bounded uniformly in $n$.

   Therefore, we should bound the integral over that part of $C_n$, where
$x_n(\lambda)\in I$. Note also that if $\Im t$ is big, then
the integral is evidently bounded by some constant (expression under the integral
decrease exponentially at the infinity), thus it suffices to bound
the integral
\[
\intd\limits_{J}
\dfrac{dt}{2\pi}\intd\limits_{C_n^{(I)}}\dfrac{d\,v}{2\pi}\exp\{\lambda^\prime(v-t)
\}\dfrac{\exp\{n(S_n(t,\lambda_0)- S_n(v,\lambda_0))\}}
{v-t},
\]
where $J$ is a finite segment of $L_n$. In view of the bound
\[
\intd\limits_J\dfrac{d\,t}{\sqrt{(x-x_0)^2+(t-y(x))^2}}\le
\sqrt{2}\intd\limits_J\dfrac{d\,t}{|x-x_0|+|t-y(x)|}\le 2\sqrt{2}\ln |x-x_0|^{-1}+C,
\]
where $x_0=x_n(\lambda_0)$, we have to estimate the integral
\begin{equation}\label{I}
\intd\limits_I(\ln |x-x_0|^{-1}+C)l^\prime(x)d\,x,
\end{equation}
where $l(x)$ is the length of the part of $C_n$ between $x_0$ and $x$.
We find from (\ref{ots_l}) that
\[
-\ln(x-x_0)\le -C\ln l(x),
\]
and, therefore, we obtain for (\ref{I})
\[
\begin{array}{c}
\intd\limits_I(\ln |x-x_0|^{-1}+C)l^\prime(x)d\,x\le
\intd\limits_I (C+\ln l(x))\,l^\prime(x)d\,x\\
=C\cdot l(x_1)-l(x_1)\ln l(x_1)\le C.
\end{array}
\]
\end{proof}
\textbf{Proof of Lemma \ref{l:11}.} Consider two cases.

1)~Let there exist a segment $\triangle=[x_n(\xi_1);x_n(\xi_2)]\subset J$,
such that $|\triangle|\ge 1/(2\ln^2 n)$ and if $x_n(\lambda)\in \triangle$,
then $|y_n(\lambda)|\ge 1/(2\ln^2 n)$.

We have from (\ref{S1})
\begin{multline}  \label{ots1}
\Re S_n(z_n(\lambda_1),\lambda_0)- \Re
S_n(z_n(\lambda_0),\lambda_0)=
\displaystyle\int\limits_{\lambda_0}^{\lambda_1}
x_n^{\prime}(\lambda)
(\lambda-\lambda_0)d\,\lambda \\
\ge\displaystyle\int\limits_{\frac{\xi_1+\xi_2}{2}}^{\xi_2}
x_n^{\prime}(\lambda) (\lambda-\lambda_0)d\,\lambda\ge
\displaystyle\frac{ (\xi_2-\xi_1)^2}{4}
\min\limits_{\lambda\in[\frac{\xi_1+\xi_2}{2};\xi_2]}
x_n^{\prime}(\lambda)
\end{multline}
According to (\ref{ab})
\begin{equation*}
x_n^\prime(\lambda)=\displaystyle\frac{a_n(\lambda)}{a^2_n(\lambda)+b^2_n(
\lambda)} \le \displaystyle\frac{1}{a_n(\lambda)}.
\end{equation*}
Using the notations (\ref{obozn}), we get from (\ref{a})
\begin{equation*}
a_n(\lambda)=2y_n^2(\lambda)\sigma_2\ge (\sqrt{2}y_n(\lambda)
\sigma_1)^2=2y^2_n(\lambda).
\end{equation*}
Hence, we have for $x_n(\lambda)\in \triangle$
\begin{equation*}
a_n(\lambda)\ge 1/(2\ln^4 n),
\end{equation*}
and
\begin{equation*}
x_n^\prime(\lambda)\le 2\ln^4 n.
\end{equation*}
Therefore,
\begin{equation*}
1/(2\ln^2 n)\le
|\triangle|=x_n(\xi_2)-x_n(\xi_1)=x^\prime_n(\theta)(\xi_2-\xi_1) \le 2\ln^4
n\cdot (\xi_2-\xi_1),
\end{equation*}
i.e.,
\begin{equation*}
\xi_2-\xi_1\ge 1/(4\ln^{6} n).
\end{equation*}
Using (\ref{obozn}) and the Schwartz inequality, we obtain
\begin{equation*}
b^2_n(\lambda)=(2y_n(\lambda)\sigma_{21})^2\le 2\sigma_{22}\cdot
2y^2_n(\lambda)\sigma_2=2a_n(\lambda)\sigma_{22}.
\end{equation*}
This, (\ref{ab}), (\ref{a}) and (\ref{cond}) yield
\begin{equation*}
x_n^\prime(\lambda)=
\displaystyle\frac{a_n(\lambda)}{a^2_n(\lambda)+b^2_n(
\lambda)}\ge
\displaystyle\frac{a_n(\lambda)}{a^2_n(\lambda)+2a_n(\lambda)
\sigma_{22}}
=\displaystyle\frac{1}{a_n(\lambda)+2\sigma_{22}}=\displaystyle
\frac{1}{2}.
\end{equation*}
Now, returning to (\ref{ots1}), we get
\begin{equation*}
\begin{array}{c}
\Re S_n(z_n(\lambda_1),\lambda_0)- \Re
S_n(z_n(\lambda_0),\lambda_0)\ge
\displaystyle\frac{(\xi_2-\xi_1)^2}{4}
\min\limits_{\lambda\in[\frac{
\xi_1+\xi_2}{2};\xi_2]} x_n^{\prime}(\lambda) \\
\ge \displaystyle\frac{(\xi_2-\xi_1)^2}{8}\ge
\displaystyle\frac{1}{ 128\ln^{12}n}
\end{array}
\end{equation*}
So the assertion of lemma is proved in this case.

2) Consider now the case when there is no segment $\triangle$,
described in the case 1. Then the segment $J$ has inside at most
$n/\ln^2 n$ of $ \{h_j^{(n)}\}_{j=1}^n$. Indeed, assume the
opposite, let $J$ have inside more than $n/\ln^2 n$ of
$\{h_j^{(n)}\}_{j=1}^n$. Split the segment $J$ into segments with
the length $1/(2\ln^2 n)$. One of these segments (denote it by
$J_1$) contains more than $n/(2\ln^4 n)$ of
$\{h_j^{(n)}\}_{j=1}^n$. Consider $\lambda$ such that
$x_n(\lambda)\in J_1$. We have for such $\lambda $ and any
$h_j^{(n)}\in J_1$
\begin{equation*}
|x_n(\lambda)-h_j^{(n)}|<1/(2\ln^2 n).
\end{equation*}
Since $J_1$ contains more than $n/(2\ln^4 n)$ of $\{h_j^{(n)}\}$, we get
from (\ref{cond})
\begin{equation*}
1=\displaystyle\frac{1}{n}\sum\limits_{j=1}^n\displaystyle\frac{1}{
(x_n(\lambda)-h_j^{(n)})^2+y^2_n(\lambda)}
\ge\displaystyle\frac{1}{2\ln^4 n}
\cdot\displaystyle\frac{1}{1/(4\ln^4 n)+y^2_n(\lambda)},
\end{equation*}
and, hence, we obtain for such $\lambda$
\begin{equation*}
|y_n(\lambda)|\ge \sqrt{1/(2\ln^4) n-1/(4\ln^4 n)}= 1/(2\ln^2 n),
\end{equation*}
which contradicts to our assumption.

Thus, the segment $J$ has inside at most $n/\ln^2 n$ of
$\{h_j^{(n)} \}_{j=1}^n$ in this case. Let us show now that there
is a $n$-independent constant $\delta$ such that
\begin{equation}  \label{3}
|\Re S_n(z_n(\lambda_1),\lambda_0)- \Re S_n(z_n(\lambda_0),\lambda_0)|\ge
\delta.
\end{equation}
Consider the function
\begin{equation}  \label{S_hat}
\widehat{S}_n(z,\lambda_0)=\displaystyle\frac{z^2}{2}+\displaystyle\frac{1}{n
}\sum\limits_{h_j^{(n)}\not\in J}\ln (z-h_j^{(n)})-\lambda_0\,z+C.
\end{equation}
We have for this function
\begin{equation*}
|\Re S_n(z_n(\lambda),\lambda_0)-\Re \widehat{S}_n(z_n(\lambda),\lambda_0)|=
\left|\displaystyle\frac{1}{2n}\sum\limits_{h_j^{(n)}\in J}\ln
((x_n(\lambda)-h_j^{(n)})^2+y^2_n(\lambda))\right|\le \displaystyle\frac{\ln
n}{2\ln^2 n}\to 0.
\end{equation*}
Therefore, it suffices to prove (\ref{3}) only for $\widehat{S}
_n(z_n(\lambda),\lambda_0)$. We know that $\Re
S_n(z_n(\lambda),\lambda_0)$ is monotone for $\lambda\in J$.
Taking into account that the difference between $\Re\widehat{S}_n$
and $\Re S_n$ converges to zero uniformly, it suffices to find two
points $x_n(\lambda)$ and $x_n(\mu)$ in $J$ such that
\begin{equation}  \label{4}
|\Re \widehat{S}_n(z_n(\lambda),\lambda_0)-\Re \widehat{S}
_n(z_n(\mu),\lambda_0)|\ge \delta
\end{equation}
for some $n$-independent $\delta$.

Replace $J$ by the segment $J^\prime$, obtained from $J$ by the exclusion of
a small $\varepsilon$-neighbor-\,hood of its endpoints. Note that
\begin{equation}  \label{pr_4}
\displaystyle\frac{d^4}{d\,x^4}\left(\Re
\widehat{S}_n(x,\lambda_0)\right)= -
\displaystyle\frac{6}{n}\sum\limits_{h_j^{(n)}\not\in
J}\displaystyle\frac{1 } {(x-h_j^{(n)})^4}.
\end{equation}
Split $J^\prime$ into three segments and choose an arbitrary
$c<1$. It is evident that the forth derivative (\ref{pr_4}) is
convex, and, hence, one can choose such third of $J^\prime$
that\newline
$\left|\displaystyle\frac{d^4}{d\,\lambda^4}\left(\Re\widehat{S}
_n(x_n(\lambda),\lambda_0)\right)\right|>c$ or
$\left|\displaystyle\frac{d^4 }{d\,\lambda^4}\left(\Re
\widehat{S}_n(x_n(\lambda),\lambda_0)\right) \right|<c$ on it.

If $\left|\displaystyle\frac{d^4}{d\,\lambda^4}\left(\Re
\widehat{S} _n(x_n(\lambda),\lambda_0)\right)\right|>c$ for this
third, then use the following elementary

\begin{proposition}
\label{p:2} Let $f$ be a $C^4[a;b]$ function. Assume that there exists a
constant $A>0$ such that
\begin{equation*}
\left|\displaystyle\frac{d^4}{d\,x^4}f(x)\right|\ge A,\,\,x\in [a;b]
\end{equation*}
Then there exist $C=C(A,|b-a|)$, and $\delta=\delta(A,|b-a|)$, and
segments $ \triangle_1,\triangle_2 \subset [a;b]$, $|\triangle_1|,
|\triangle_2|>\delta$ such that for any $x_1\in \triangle_1$,
$x_2\in \triangle_2$
\begin{equation*}
|f(x_1)-f(x_2)|\ge C.
\end{equation*}
\end{proposition}

Since $\hat{S}$ satisfies the condition of the proposition, there
exist $ \triangle_1,\triangle_2\subset J^\prime$ such that we have
for any $x_1\in \triangle_1$ and any $x_2\in \triangle_2$
\begin{equation}  \label{5}
|\Re \widehat{S}_n(x_1,\lambda_0)-\Re \widehat{S}_n(x_2,\lambda_0)|\ge
\delta.
\end{equation}
It is easy to see that both $\triangle_1$ and $\triangle_2$
contain $ x_n(\lambda)$ for which the corresponding $y_n(\lambda)$
obeys the inequality $y<1/\ln^4 n$ (or we have the case 1). We
obtain for these points
\begin{equation*}
\begin{array}{c}
|\Re \widehat{S}_n(z_n(\lambda),\lambda_0)-\Re \widehat{S}
_n(x_n(\lambda),\lambda_0)|=\displaystyle\frac{1}{n}\sum\limits_{h_j^{(n)}
\not\in J}\ln \left(1+\displaystyle\frac{y^2_n(\lambda)}{(x_n(
\lambda)-h_j^{(n)})^2}\right) \\
\le\displaystyle\frac{1}{n}\sum\limits_{h_j^{(n)}\not\in
J}\displaystyle
\frac{y^2_n(\lambda)}{(x_n(\lambda)-h_j^{(n)})^2}\le
1/(\varepsilon\cdot\ln^8 n)
\end{array}
\end{equation*}
This and (\ref{5}) imply (\ref{4}), thus (\ref{3}).

If $\left|\displaystyle\frac{d^4}{d\,\lambda^4}\left(\Re
\widehat{S} _n(x_n(\lambda),\lambda_0)\right)\right|<c$, then we
consider the second derivative
\begin{equation}  \label{pr_2}
\displaystyle\frac{d^2}{d\,\lambda^2}\left(\Re \widehat{S}
_n(x_n(\lambda),\lambda_0)\right)= 1-\displaystyle\frac{1}{n}
\sum\limits_{h_j^{(n)}\not\in J}\displaystyle\frac{1} {(x_n(
\lambda)-h_j^{(n)})^2}.
\end{equation}
Note that
\begin{equation*}
\left(\displaystyle\frac{1}{n}\sum\limits_{h_j^{(n)}\not\in
J}\displaystyle \frac{1} {(x_n(\lambda)-h_j^{(n)})^2}\right)^2\le
\displaystyle\frac{1}{n} \sum\limits_{h_j^{(n)}\not\in
J}\displaystyle\frac{1}{(x_n( \lambda)-h_j^{(n)})^4}\le c/6.
\end{equation*}
This and (\ref{pr_2}) yield
\begin{equation*}
\left|\displaystyle\frac{d^2}{d\,\lambda^2}(\Re \widehat{S}
_n(x_n(\lambda),\lambda_0))\right|\ge 1-\sqrt{c/6}.
\end{equation*}
This bound implies (\ref{3}) by the same argument as in
Proposition \ref{p:2} . Thus, since the condition (\ref{2}) is
more weak than the condition (\ref {3}), we have proved (\ref{2})
in any case.$\quad\Box$

\medskip  Note that according to the Hadamard inequality
\begin{multline}
\det
\left\{\displaystyle\frac{1}{n}K_n\left(\lambda_0+\displaystyle\frac{x_i
}{n}, \lambda_0+\displaystyle\frac{x_j}{n}\right)\right\}_{i,j=1}^{m} \\
\le\prod\limits_{i=1}^m\left(\sum\limits_{j=1}^m\displaystyle\frac{1}{n}
K_n\left(\lambda_0+\displaystyle\frac{x_i}{n},
\lambda_0+\displaystyle\frac{ x_j}{n}\right)
\displaystyle\frac{1}{n}K_n\left(\lambda_0+\displaystyle\frac{
x_j}{n},\lambda_0+\displaystyle\frac{x_i}{n}\right)\right)^{1/2}.
\end{multline}
Since the second marginal density is positive,
\begin{equation*}
K_n(x,y)K_n(y,x)\le K_n(x,x) K_n(y,y).
\end{equation*}
According to Lemma \ref{l:9} this means that
\begin{equation*}
\det
\left\{\displaystyle\frac{1}{n}K_n\left(\lambda_0+\displaystyle\frac{x_i
}{n},
\lambda_0+\displaystyle\frac{x_j}{n}\right)\right\}_{i,j=1}^{m}\le
m^{m/2}C^m,
\end{equation*}
and, hence, the integral over the complement of
$\Omega_\varepsilon$ in (\ref {limr}) can be bounded by $C
\delta$. Since we can take $\delta$ small arbitrary, the condition
(\ref{Un}) is proved.

\section{Appendix.}

  We present here certain facts of the Grassmann variables and the
Grassmann integration. An introduction to this theory is given in
\cite{Be:87} and \cite{Ef:97}, and in this section we will follow
to these books.

\subsection{Grassmann algebra $\Lambda$.}
Let us consider the set of formal variables $\{\psi_j\}_{j=1}^n$,
which satisfy the following anticommutation conditions
\begin{equation*}
\psi_j\psi_k+\psi_k\psi_j=0,\quad j,k=\overline{1,n}.
\end{equation*}
In particular, for $k=j$ we obtain
\begin{equation*}
\psi_j^2=0.
\end{equation*}
To any variable $\psi_j$ we put into correspondence another
variable $ \overline{\psi}_j$, which we call \textit{the
conjugate}\, of $\psi_j$. We assume that these conjugate variables
$\{\overline{\psi}_j\}_{j=1}^n$ also anticommute with each others
and with $\{\psi_j\}_{j=1}^n$:
\begin{equation*}
\overline{\psi}_j\psi_k+\psi_k\overline{\psi}_j=\overline{\psi}_j
\overline{ \psi}_k+\overline{\psi}_k\overline{\psi}_j=0.
\end{equation*}
These two sets of variables $\{\psi_j\}_{j=1}^n$ and
$\{\overline{\psi} _j\}_{j=1}^n$ generate the Grassmann algebra
$\Lambda$. Taking into account that $\psi_j^2=0$, we have that all
elements of $\Lambda$ are some polynomials of $\{\psi_j\}$ and
$\{\overline{\psi}_j\}$. One can extend the operation of
conjugation to the whole $\Lambda$ by setting
\begin{equation*}
\overline{\alpha\psi}=\overline{\alpha}\overline{\psi},\quad
\overline{ \overline{\psi}}=-\psi,\quad
\overline{\psi_1\psi_2}=\overline{\psi_1} \overline{\psi_2}.
\end{equation*}
We can also define functions of Grassmann variables. Let $\chi$ be
some element of $\Lambda$. For any analytical function $f$ by
$f(\chi)$ we mean the element of $\Lambda$ obtained by
substituting $\chi$ in the Taylor series of $f$ near zero. Since
$\chi$ is a polynomial of $\{\psi_j\}$, $\{ \overline{\psi}_j\}$,
there exists such $l$ that $\chi^l=0$, and hence the series
terminates after a finite number of terms and so $f(\chi)\in
\Lambda$.

Let us also call by a \textit{numerical part} of some function of
Grassmann's elements its value obtained by putting all $\psi_j$
and $ \overline{\psi}_j$ formally equal to zero (in other word,
the first coefficient of Taylor series).

\subsection{Linear algebra over $\Lambda$}

A super-vector of the first type is defined as a $(n+m)$ dimensional
vector-column whose first $m$ coordinates $\{\chi_j\}_{j=1}^m$ are
anticommuting elements of $\Lambda$ (i.e.,an elements containing only terms
of odd power) and the last $n$ coordinates $\{s_j\}_{j=1}^n$ are commuting
ones (i.e.,elements containing only terms of even power):
\begin{equation*}
\Phi_1=(\chi_{1},\ldots,\chi_{m},s_{1},\ldots,s_{n})^t.
\end{equation*}
One can also consider super-vectors of the second type: a $(m+n)$
dimensional vector-column whose first $m$ coordinates $\{s_j\}_{j=1}^m$ are
commuting elements and the last $n$ coordinates $\{\chi_j\}_{j=1}^n$ are
anticommuting ones:
\begin{equation*}
\Phi_2=(s_{1},\ldots,s_{m},\chi_{1},\ldots,\chi_{n})^t.
\end{equation*}
The Hermitian conjugate $\Phi^+$ is given by the following expression:
\begin{equation*}
\Phi_1^+=(\overline{\chi}_1,\ldots,\overline{\chi}_m,\overline{s}_1,\ldots,
\overline{s}_n),\quad
\Phi_2^+=(\overline{s}_1,\ldots,\overline{s}_m,
\overline{\chi}_1,\ldots, \overline{\chi}_n)
\end{equation*}
Super-vectors of each type obviously form a linear space. A linear
transformation in these spaces are realized by super-matrices:
\begin{equation*}
\widetilde{\Phi}=F\,\Phi,\quad F=\left(
\begin{array}{cc}
a & \sigma \\
\rho & b \\
\end{array}
\right),\quad
\end{equation*}
where $a$ and $b$ are $n\times n$ and $m\times m$ matrices
containing only commuting elements of algebra, $\sigma$ and $\rho$
are $n\times m$ and $ m\times n$ matrices containing only
anticommuting ones.

Two super-matrices $F$ and $G$ can be multiplied in a usual way
\begin{equation*}
(F\,G)_{j,k}=\sum\limits_{l=1}^{m+n}F_{j,l}G_{l,k}.
\end{equation*}
Now let us define super-analogs of traces and determinants of matrices.
\begin{equation*}
\hbox{str}\, F=\hbox{Tr}\, a-\hbox{Tr}\, b,\quad \hbox{sdet}\, F=
\displaystyle\frac{\det\,(a-\sigma\, b^{-1}\,\rho)}{\det\, b}.
\end{equation*}
These definitions look very unusual but they allow us to preserve
some basic properties of traces and determinants (see \cite{Ef:97}):
\begin{equation*}
\hbox{str}\,(FG)=\hbox{str}\,(GF),\quad
\hbox{sdet}\,(FG)=\hbox{sdet}\,F\cdot \hbox{sdet}\,G,\quad
\ln\hbox{sdet}\, F=\hbox{str}\,\ln \, F.
\end{equation*}
Super-analog of Hermitian conjugation of matrices can be defined as
\begin{equation*}
F^+=\left(
\begin{array}{cc}
a^+ & -\rho^+ \\
\sigma^+ & b^+ \\
\end{array}
\right),\quad (FG)^+=G^+F^+,\quad (F^+)^+=F.
\end{equation*}
According to this definition one can introduce a Hermitian and unitary super
matrices. The Hermitian super-matrix $F$ satisfies the condition $F^+=F$
while the unitary super-matrix $F$ satisfies the condition $F^+\,F=F\,F^+=1$.

Similarly to ordinary matrices, Hermitian super-matrices can be
diagonalized by unitary super-matrices (see also \cite{Ef:97}).

Indeed, an arbitrary Hermitian super-matrix has the form
\begin{equation*}
F=\left(
\begin{array}{cc}
a & \sigma \\
\sigma^+ & b \\
\end{array}
\right),\quad
\end{equation*}
where $a$ and $b$ are $n\times n$ Hermitian matrices containing only
commuting elements of algebra and $\sigma$ is a $n\times n$ matrix
containing only anticommuting ones. Suppose that all numerical parts of the
eigenvalues of the matrices $a$ and $b$ are distinct (the eigenvalues of
matrices containing only commuting elements can be defined by the same way
as for ordinary matrices. The way to find roots of the characteristic
polynomial is described below). Find such commuting elements $\lambda$ that
\begin{equation*}
F\,\left(
\begin{array}{c}
S \\
\chi \\
\end{array}
\right)=\lambda \left(
\begin{array}{c}
S \\
\chi \\
\end{array}
\right)
\end{equation*}
or
\begin{equation}  \label{Ur_S_Xi}
\left\{
\begin{array}{c}
(a-\lambda)S+\sigma\chi=0 \\
\sigma^+S+(b-\lambda)\chi=0 \\
\end{array}
\right..
\end{equation}

Excluding $\chi$, we get the system of linear equations for $
S=(s_1,\ldots,s_n)$:
\begin{equation}  \label{Ur_S}
((a-\lambda)-\sigma\,(b-\lambda)^{-1}\sigma^+)\, S=0.
\end{equation}
If $\det ((a-\lambda)-\sigma\,(b-\lambda)^{-1}\sigma^+)=0$, then
the system ( \ref{Ur_S}) has a nontrivial solution, i.e., some
solution with a nonzero numerical part. Indeed, consider a maximum
minor of the matrix $
C(\lambda)=(a-\lambda)-\sigma\,(b-\lambda)^{-1}\sigma^+$ with a
nonzero numerical part. Since $\hbox{rank}\,(a-\lambda)\ge n-1$
(because all eigenvalues of $a$ are distinct) and $\det
C(\lambda)=0$, the rank of this minor is $(n-1)$. Without loss of
generality we can assume that it is an upper right minor. The last
equation of the system can be omitted, and the first $(n-1)$ one
can be solved with respect to $s_1,s_2,\ldots,s_{n-1}$ with a
parameter $s_n$ by using the Kramer rule (since a numerical part
of the main determinant is nonzero, we can divide by it). Taking
arbitrary $ s_n\not= 0$, we obtain a nontrivial solution.

Thus, if $\det C(\lambda)=0$, then the system (\ref{Ur_S}) has a nontrivial
solution. Having this solution, one can construct
\begin{equation}  \label{Xi}
\chi=-(b-\lambda)^{-1}\sigma^+S,\quad \chi^+=-S^+\sigma(b-\lambda)^{-1},
\end{equation}
which represents a solution of the system (\ref{Ur_S_Xi}). Choosing a
constant, we can obtain a normalized solution of the system, i.e., the
solution $\Phi=(S,\,\chi)^t$ such that $\Phi^+\Phi=S^+S+\chi^+\chi~=~1$.

Hence, we should find the solutions of the equation
\begin{equation}  \label{e_det}
\hbox{det}((a-\lambda)-\sigma\,(b-\lambda)^{-1}\sigma^+)~=~0,
\end{equation}
i.e.,the roots of some polynomial (denote it by $f(x)$) whose coefficients
are elements of $\Lambda$. Let us seek these roots by the Newton method
using the eigenvalues of the matrix $a$ as a zero approximation.

Let
\begin{equation*}
x_1=\lambda_0,\quad x_n=x_{n-1}-\displaystyle\frac{f(x_{n-1})}{
f^\prime(x_{n-1})}.
\end{equation*}
It can be prove by induction that $f(x_n)=f(x_1)^n\cdot g(x_n)$
and the numerical part of $f(x_1)$ is zero. Since there exists $N$
such that $ f^N(x_1)=0$, for $n>N$ $f(x_n)=0$. This means that for
$n>N$ $x_n=x_N$ and so $x_N$ is the solution of $f(x)=0$,
corresponding to $\lambda$.

In such a way we find $n$ eigenvalues and the normalized eigenvectors of the
type $\Phi=(S,\,\chi)^t$, corresponding to these eigenvalues.

Similarly we can find the eigenvectors of the type $\Phi=(\chi,\,S)^t$, but
in this case we should use the eigenvalues of the matrix $b$ instead of $a$.
It is easy to show that the eigenvectors corresponding to the distinct
eigenvalues are orthogonal to each other. So constructing the super-matrix
from all these vectors, we obtain the unitary matrix $U$, diagonalizing $F$.

As an example consider the case $n=1$. In this case we have
\begin{equation*}
F=\left(
\begin{array}{cc}
a & \sigma \\
\overline{\sigma} & b \\
\end{array}
\right),\quad
\end{equation*}
where $a$ and $b$ are distinct real numbers, $\sigma$ is some anticommuting
element of $\Lambda$.

Equation (\ref{e_det}) has the form
\begin{equation*}
(a-\lambda)-\displaystyle\frac{\sigma\,\overline{\sigma}}{b-\lambda}
=~0,\quad f(x)=(a-\lambda)(b-\lambda)-\sigma\,\overline{\sigma}.
\end{equation*}
As a zero approximation we should take $a$:
\begin{equation*}
f(a)=-\sigma\,\overline{\sigma},\quad f^\prime(a)=a-b.
\end{equation*}
Hence,
\begin{equation*}
x_2=a+\displaystyle\frac{\sigma\,\overline{\sigma}}{a-b},
\end{equation*}
and therefore
\begin{equation*}
f(x_2)=-\displaystyle\frac{\sigma\,\overline{\sigma}}{a-b}(b-a-
\displaystyle
\frac{\sigma\,\overline{\sigma}}{a-b})-\sigma\,\overline{\sigma}=0.
\end{equation*}
Thus, one of the eigenvalues is
\begin{equation*}
\lambda_1=a+ \displaystyle\frac{\sigma\,\overline{\sigma}}{a-b}.
\end{equation*}
Find the normalized eigenvector of the type $\Phi=(S,\,\chi)^t$,
corresponding to this eigenvalue. The system (\ref{Ur_S}) in this case is
degenerated, and so $S=(s_1)$ is arbitrary. From (\ref{Xi}) we obtain that
\begin{equation*}
\chi=\displaystyle\frac{\overline{\sigma}\,S}{a-b},\quad \chi^+=-
\displaystyle\frac{\sigma\,S}{a-b}.
\end{equation*}
Hence
\begin{equation*}
S^2+\chi^+\chi=S^2(1-\displaystyle\frac{\overline{\sigma}\sigma}{(a-b)^2}).
\end{equation*}
So to normalize the vector we should take
\begin{equation*}
S=1+\displaystyle\frac{\overline{\sigma}\sigma}{2(a-b)^2}.
\end{equation*}
Thus, the eigenvector corresponding to the eigenvalue $\lambda_1$ has the
form
\begin{equation*}
\Phi_1=\left(
\begin{array}{c}
1+\displaystyle\frac{\overline{\sigma}\sigma}{2(a-b)^2} \\
\displaystyle\frac{\overline{\sigma}}{a-b} \\
\end{array}
\right)
\end{equation*}
Similarly the second eigenvalue (corresponding to $b$) is
\begin{equation*}
\lambda_2=b+ \displaystyle\frac{\sigma\,\overline{\sigma}}{a-b},
\end{equation*}
and the eigenvector corresponding to this eigenvalue has the form
\begin{equation*}
\Phi_2=\left(
\begin{array}{c}
-\displaystyle\frac{\sigma}{a-b} \\
1-\displaystyle\frac{\overline{\sigma}\sigma}{2(a-b)^2} \\
\end{array}
\right).
\end{equation*}
Constructing the super-matrix
\begin{equation*}
U=\left(
\begin{array}{cc}
1+\displaystyle\frac{\overline{\sigma}\sigma}{2(a-b)^2} &
-\displaystyle
\frac{\sigma}{a-b} \\
\displaystyle\frac{\overline{\sigma}}{a-b} &
1-\displaystyle\frac{\overline{
\sigma}\sigma}{2(a-b)^2} \\
\end{array}
\right),
\end{equation*}
from these vectors we get that
\begin{equation*}
U^+FU=\left(
\begin{array}{cc}
a+\displaystyle\frac{\sigma\overline{\sigma}}{a-b} & 0 \\
0 & b+\displaystyle\frac{\sigma\overline{\sigma}}{a-b} \\
\end{array}
\right).
\end{equation*}

\subsection{Integral over $\Lambda$.}

Following Berezin\cite{Be:87}, we define the operation of
integration with respect to the anticommuting variables in a
formally way:
\begin{equation*}
\displaystyle\int d\,\psi_j=\displaystyle\int
d\,\overline{\psi}_j=0,\quad \displaystyle\int
\psi_jd\,\psi_j=\displaystyle\int \overline{\psi}_jd\,
\overline{\psi}_j=1.
\end{equation*}
This definition can be extend on the general element of $\Lambda$ by the
linearity. A multiple integral is defined to be repeated integral. The
"differentials" $d\,\psi_j$ and $d\,\overline{\psi}_k$ anticommute with each
other and with the variables $\psi_j$ and $\overline{\psi}_k$.

Therefore, if
\begin{equation*}
f(\chi_1,\ldots,\chi_m)=a_0+\sum\limits_{j_1=1}^m
a_{j_1}\chi_{j_1}+\sum\limits_{j_1<j_2}a_{j_1j_2}\chi_{j_1}\chi_{j_2}+
\ldots+a_{1,2,\ldots,m}\chi_1\ldots\chi_m,
\end{equation*}
then
\begin{equation*}
\displaystyle\int f(\chi_1,\ldots,\chi_m)d\,\chi_m\ldots
d\,\chi_1=a_{1,2,\ldots,m}.
\end{equation*}

Let now $f=f(X,\chi)$, where $\chi=(\chi_1,\ldots,\chi_m)$ is a vector of
the anticommuting elements of $\Lambda$, end $X=(x_1,\ldots,x_n)$ is a
vector of the commuting ones. Let $y_i$ be a numerical part of $x_i$. Then
\begin{equation*}
\displaystyle\int\limits \displaystyle\int f(X,\chi)d\,x_1\ldots
d\,x_nd\,\chi_m\ldots d\,\chi_1=
\displaystyle\int\limits_{\widetilde{U}} \displaystyle\int
f(Y,\chi)d\,y_1\ldots d\,y_nd\,\chi_m\ldots d\,\chi_1,
\end{equation*}
where $\widetilde{U}$ is a domain, where coordinates $Y=(y_1,\ldots,y_n)$
vary, and integral over $\widetilde{U}$ is  a usual Lebesgues integral.

Let $A$ be an ordinary Hermitian matrix. The following Gaussian integral is
well-known
\begin{equation}  \label{G_C}
\displaystyle\int
\exp\{-\sum\limits_{j,k=1}^nA_{j,k}z_j\overline{z}_k\}
\prod\limits_{j=1}^n\displaystyle\frac{d\,\Re z_jd\,\Im z_j}{\pi}=
\displaystyle\frac{1}{\det A}.
\end{equation}
One of the most important formulas of the super-symmetry method is
an analog of formula (\ref{G_C}) for Grassmann variables
\cite{Be:87}:
\begin{equation}  \label{G_Gr}
\displaystyle\int
\exp\{-\sum\limits_{j,k=1}^nA_{j,k}\overline{\psi} _j\psi_k\}
\prod\limits_{j=1}^nd\,\overline{\psi}_jd\,\psi_j=\det A.
\end{equation}
Combining these two formulas, we obtain another important one: if $F$ is a
Hermitian super-matrix and $\Phi=(X,\chi)^t$ is a super-vector, then
\begin{equation}  \label{G_comb}
\displaystyle\int \exp\{-\Phi^+F\Phi\}d\,\Phi^+d\,\Phi=\hbox{sdet}^{-1}\, F,
\end{equation}
where
\begin{equation*}
d\,\Phi^+ d\,\Phi=\prod\limits_{j=1}^m\overline{\chi}_j\chi_j
\prod\limits_{j=1}^n\displaystyle\frac{\Re x_j\Im x_j}{\pi}.
\end{equation*}

\subsection{Derivatives with respect to anticommuting variables.}

Let us define the left and the right derivatives with respect to
anticommuting variables. Since any element of the algebra $\Lambda$ is a
polynomial of $\{\psi_j\}$ and $\{\overline{\psi}_j\}$, it is sufficient to
define derivatives only for monomials and then extend by the linearity.

We define the left derivative as (see \cite{Be:87}):
\begin{equation*}
\displaystyle\frac{\partial}{\partial\chi_j}\chi_{i_1}\ldots\chi_{i_k}=
\left\{
\begin{array}{ll}
0, & \quad i_1,\ldots,i_k\not=j, \\
(-1)^{s-1}\chi_{i_1}\ldots\chi_{i_{s-1}}\chi_{i_{s+1}}\ldots\chi_{i_k}, &
\quad i_s=j. \\
\end{array}
\right.
\end{equation*}

The right derivative differs from the left one by sign:
\begin{equation*}
\chi_{i_1}\ldots\chi_{i_k}\displaystyle\frac{\partial}{\partial\chi_j}=
\left\{
\begin{array}{ll}
0, & \quad i_1,\ldots,i_k\not=j, \\
(-1)^{k-s}\chi_{i_1}\ldots\chi_{i_{s-1}}\chi_{i_{s+1}}\ldots\chi_{i_k}, &
\quad i_s=j. \\
\end{array}
\right.
\end{equation*}
Note that for the odd elements the left and the right derivatives
are equal and so in this case we can use the usual notation
$\displaystyle\frac{
\partial f} {\partial \chi}$.

\subsection{Change of variables in integrals.}

Consider the integral
\begin{equation}  \label{I1}
\displaystyle\int\limits_U \displaystyle\int f(X,\chi)d\,\chi d\,X,
\end{equation}
where $X=(x_1,\ldots,x_n)$ are commuting variables whose numerical parts
vary in the domain $U$, and $\chi=(\chi_1,\ldots,\chi_m)$ are anticommuting
ones.

Change of variables in the integral (\ref{I1}) is a transformation from one
system of generators of $\Lambda$ to another one preserving the evenness
\begin{equation}  \label{N_per}
x_i=x_i(Y,\eta),\quad \chi_i=\chi_i(Y,\eta),
\end{equation}
where $Y=(y_1,\ldots,y_n)$ are commuting variables, whose numerical parts
vary in the domain $\widetilde{U}$, and $\eta=(\eta_1,\ldots,\eta_m)$ are
anticommuting ones.

Change of variables in an ordinary integral leads to the appearance of the
Jacobian which is equal to the determinant of the partial derivatives
matrix. For the super-integrals the situation is similar.

Let $f$ be a finite function in the domain $U$, i.e., $\hbox{supp}
f$ (with respect to the numerical part of the vector $X$) is
inside the domain $U$. Then (see \cite{Be:87})
\begin{equation}  \label{Zam}
\displaystyle\int\limits_U \displaystyle\int f(X,\chi)d\,\chi
d\,X= \displaystyle\int\limits_{\widetilde{U}} \displaystyle\int
f(X(Y,\eta),\chi(Y,\eta))\triangle(\{X,\chi\}/\{Y,\eta\})d\,\chi
d\,X,
\end{equation}
where
\begin{equation}  \label{Ber}
\triangle(\{X,\chi\}/\{Y,\eta\})=\hbox{sdet} R,\quad R=\left(
\begin{array}{cc}
a & \alpha \\
\beta & b \\
\end{array}
\right),
\end{equation}

\begin{equation*}
\begin{array}{cc}
a_{ik}=\displaystyle\frac{\partial x_i}{\partial y_k}\quad &
\alpha_{ik}= x_i
\displaystyle\frac{\partial}{\partial\eta_k} \\
\beta_{ik}=\displaystyle\frac{\partial\chi_i}{\partial y_k}\quad &
b=
\displaystyle\frac{\partial\chi_i}{\partial \eta_k} \\
\end{array}
\end{equation*}
The function $\triangle(\{X,\chi\}/\{Y,\eta\})$ is often called
\textit{ Berezinian} of the change (\ref{N_per}).

Note that differently from the ordinary integrals, in the case of
super-integrals if $f$ is not a finite function in the domain $U$, then
formula (\ref{Zam}) is not correct. There are some extra terms appearing in
it.

Let the domain $U$ be defined by the condition $u(X)>0$ for some
function $u$. Denote by $v(Y,\eta)$ the function $u(X(Y,\eta))$,
and let $v(Y)$ be the numerical part of $v(Y,\eta)$. In new
coordinates the domain $U$ will be defined by the condition
$v(Y)>0$ and (see \cite{Be:87})
\begin{multline}  \label{Zam1}
\displaystyle\int\limits_U \displaystyle\int f(X,\chi)d\,\chi
d\,X= \displaystyle\int\limits_{\widetilde{U}} \displaystyle\int
f(X(Y,\eta),\chi(Y,\eta))\triangle(\{X,\chi\}/\{Y,\eta\})d\,\chi d\,X \\
+\displaystyle\int f(X(Y,\eta),\chi(Y,\eta))
\triangle(\{X,\chi\}/\{Y,\eta\})\delta(v(Y))(v(Y,\eta)-v(Y))d\,\chi
d\,X+\ldots,
\end{multline}
where dots means the sum of terms containing $\delta^{(k)}(v(Y))$ under the
integral, i.e., all extra terms are integrals along the boundary of the
domain $U$.

Let
\begin{equation*}
F=\left(
\begin{array}{cc}
a & \sigma \\
\sigma^+ & ib \\
\end{array}
\right),\quad G=\left(
\begin{array}{cc}
c & \eta \\
\eta^+ & id \\
\end{array}
\right).
\end{equation*}

Then $F$ and $G$ can be diagonalized, i.e., there exist unitary
super-matrices $U$ and $V$ such that
\begin{equation*}
\begin{array}{cc}
F=U^{-1}SU,\,\hbox{где}\quad
S=\hbox{diag}\,(s_{11},\ldots,s_{1m},is_{21},
\ldots,s_{2m}), &  \\
G=V^{-1}RV,\,\hbox{где}\quad
R=\hbox{diag}\,(r_{11},\ldots,r_{1m},ir_{21}, \ldots,r_{2m}).
\end{array}
\end{equation*}
Consider the integral
\begin{equation*}
2^{m(m-1)}\displaystyle\int
\exp\left(-\displaystyle\frac{1}{2t}\hbox{str}
(F-G)^2\,\right)d\,G,
\end{equation*}
where
\begin{equation*}
d\,G=\displaystyle\frac{1}{\pi^{m^2}}\prod\limits_{j=1}^md\,c_{j,j}d
\,d_{j,j}\prod\limits_{j<k}d\,\Re c_{j,k} d\,\Im c_{j,k}d\,\Re
d_{j,k}d\,\Im
d_{j,k}\prod\limits_{j,k=1}^md\,\overline{\eta}_{j,k}d\,\eta_{j,k}.
\end{equation*}

If we make the change $G=V^{-1}RV$, the differential $d\,G$ will
transform into the form (see \cite{Gu:91})
\begin{equation*}
d\,G=B_m(R)^2d\,Rd\,\mu(V)
\end{equation*}
where $d\,R=d\,r_{11}\ldots d\,r_{2m}$, $d\,\mu(V)$ is the Haar measure of
the group of unitary super-matrices, and $B_m(R)^2$ is a Berezinian of this
change, which equals to the square of the Cauchy determinant
\begin{equation}  \label{B}
B_m(R)=\det\left[\displaystyle\frac{1}{r_{1j}-ir_{2k}}\right].
\end{equation}

We will use also the generalization of the
Harish-Chandra/Itzykson-Zuber formula for the case of Grassmann
variables. Let us recall that the Harish-Chandra/Itzykson-Zuber
formula has the form (see, for example, \cite{Me:91}):
\begin{equation*}
\displaystyle\int\exp\{\hbox{Tr}\,
AU^*BU\}d\,U=\displaystyle\frac{
\det\{\exp(a_ib_j)\}}{\triangle(A)\triangle(B)},
\end{equation*}
where $A$, $B$ are Hermitian matrices, $a_i$, $b_j$ are their
eigenvalues, $ d\,U$ is an integration over the group of unitary
matrices, and $\triangle(A) $ is the Vandermonde determinant
constructing of the eigenvalues of the matrix $A$, i.e.,
\begin{equation*}
\triangle(A)=\prod\limits_{i<j}(a_i-a_j)
\end{equation*}
The super-analog of this formula has the form (see \cite{Gu:91}):
\begin{multline}  \label{Itz_Z_gr}
2^{m(m-1)}\displaystyle\int
\exp\left(-\displaystyle\frac{1}{2t}\hbox{str}
(F-G)^2\,\right)d\,\mu(V) \\
=(1-\eta(S))\displaystyle\frac{\delta(R)}{B_m^2(R)}+\displaystyle\frac{1}{
(2\pi
t)^m}\displaystyle\frac{\exp\left(-\displaystyle\frac{1}{2t}\hbox{str}
(S-R)^2\,\right)} {B_m(S)B_m(R)},
\end{multline}
where
\begin{equation*}
\eta(S)= \left\{
\begin{array}{cl}
0, & \hbox{if any two}\,\, s_{1j}=0,\, s_{2k}=0, \\
1, & \hbox{otherwise}.
\end{array}
\right.
\end{equation*}

{\bf Acknowledgements.}
 The author  is grateful to Prof.L.A.Pastur
 for  statement of the problem and fruitful discussion.

\end{document}